\documentclass[reqno,11pt]{amsart}
\usepackage[utf8]{inputenc}
\usepackage{graphicx}
\usepackage{slashed}
\usepackage{amscd}
\usepackage{amssymb}
\usepackage{esint}
\usepackage{enumitem}
\usepackage{dsfont}
\usepackage[mathscr]{eucal}
\usepackage{upgreek}
\usepackage{pst-all}

\numberwithin{equation}{section}
\allowdisplaybreaks[1]

\title[Banach Manifold Structure of Causal Fermion Systems]{Banach Manifold Structure and Infinite-Dimensional Analysis for Causal Fermion Systems}

\author[F.\ Finster]{Felix Finster}

\author[M.\ Lottner]{Magdalena Lottner \\ \\  January 2021}
\address{Fakult\"at f\"ur Mathematik  Universit\"at Regensburg  D-93040 Regensburg  Germany}
\email{finster@ur.de, magdalena.lottner@ur.de}

\newtheorem{Def}{Definition}[section]
\newtheorem{Thm}[Def]{Theorem}
\newtheorem{Prp}[Def]{Proposition}
\newtheorem{Lemma}[Def]{Lemma}
\newtheorem{Remark}[Def]{Remark}

\newcommand{\Thanks}{\vspace*{.5em} \noindent \thanks}
\newcommand{\beq}{\begin{equation}}
\newcommand{\eeq}{\end{equation}}
\newcommand{\Proof}{\begin{proof}}
	\newcommand{\QED}{\end{proof} \noindent}
\newcommand{\QEDrem}{\ \hfill $\Diamond$}

\newcommand{\la}{\langle}
\newcommand{\ra}{\rangle}

\newcommand{\Sl}{\mbox{$\prec \!\!$ \nolinebreak}}
\newcommand{\Sr}{\mbox{\nolinebreak $\succ$}}

\newcommand{\C}{\mathbb{C}}
\newcommand{\R}{\mathbb{R}}
\newcommand{\1}{\mbox{\rm 1 \hspace{-1.05 em} 1}}

\newcommand{\N}{\mathbb{N}}

\DeclareMathOperator{\tr}{tr}

\renewcommand{\L}{{\mathcal{L}}}
\newcommand{\Sact}{{\mathcal{S}}}

\newcommand{\U}{\text{\rm{U}}}

\renewcommand{\H}{\mathscr{H}}
\newcommand{\Shil}{\mycal{S}}

\newcommand{\Lin}{\text{\rm{L}}}
\newcommand{\F}{{\mathscr{F}}}
\newcommand{\reg}{\text{\rm{reg}}}

\DeclareMathOperator{\Symm}{\mbox{\rm{Symm}}}

\DeclareMathOperator{\re}{Re}

\DeclareMathOperator{\supp}{supp}

\newcommand{\itemD}{\item[{\raisebox{0.125em}{\tiny $\blacktriangleright$}}]}

\newcommand{\J}{\mathfrak{J}}
\newcommand{\s}{\mathfrak{s}}
\newcommand{\Jdiff}{\mathfrak{J}^\text{\rm{\tiny{diff}}}}
\newcommand{\Jtest}{\mathfrak{J}^\text{\rm{\tiny{test}}}}

\newcommand{\Gdiff}{\Gamma^\text{\rm{\tiny{diff}}}}
\newcommand{\Gtest}{\Gamma^\text{\rm{\tiny{test}}}}

\newcommand{\Ctest}{C^\text{\rm{\tiny{test}}}}
\newcommand{\fu}{\mathfrak{u}}
\newcommand{\fv}{\mathfrak{v}}

\newcommand{\bitem}{\begin{itemize}[leftmargin=2.5em]}
\newcommand{\eitem}{\end{itemize}}

\newcommand{\E}{\mathscr{E}}
\newcommand{\id}{\text{id}}
\newcommand{\symm}{\text{\tiny{symm}}}

\newcommand{\bu}{{\mathbf{u}}}
\newcommand{\bv}{\mathbf{v}}
\newcommand{\bw}{\mathbf{w}}

\DeclareFontFamily{OT1}{rsfso}{}
\DeclareFontShape{OT1}{rsfso}{m}{n}{ <-7> rsfso5 <7-10> rsfso7 <10-> rsfso10}{}
\DeclareMathAlphabet{\mycal}{OT1}{rsfso}{m}{n}

\begin{document}
\maketitle
\begin{abstract}
A mathematical framework is developed for the analysis of causal fermion systems in the
infinite-dimensional setting.
It is shown that the regular spacetime point operators form a Banach manifold endowed with a canonical
Fr\'echet-smooth Riemannian metric.
The so-called expedient differential calculus is introduced with the purpose of treating
derivatives of functions on Banach spaces which are differentiable only in certain directions.
A chain rule is proven for H\"older continuous functions which are differentiable
on expedient subspaces.
These results are made applicable to causal fermion systems by proving 
that the causal Lagrangian is H\"older continuous.
Moreover, H\"older continuity is analyzed for the integrated causal Lagrangian.
\end{abstract}
\tableofcontents

\section{Introduction}
The theory of {\em{causal fermion systems}} is a recent approach to fundamental physics
(see the basics in Section~\ref{secprelim}, the reviews~\cite{ nrstg, review, dice2014}, the textbook~\cite{cfs}
or the website~\cite{cfsweblink}).
In this approach, spacetime and all objects therein are described by a measure~$\rho$
on a set~$\F$ of linear operators of rank at most~$2n$ on a Hilbert space~$(\H, \la .|. \ra_\H)$. 
The physical equations are formulated via the so-called {\em{causal action principle}},
a nonlinear variational principle where an action~$\Sact$ is minimized under variations of the measure~$\rho$.
If the Hilbert space~$\H$ is {\em{finite-dimensional}}, the set~$\F$ is a locally compact topological
space. Making essential use of this fact, it was shown in~\cite{continuum} that the causal action principle
is well-defined and that minimizers exist. Moreover, as is worked out in detail in~\cite{gaugefix},
the interior of~$\F$ (consisting of the so-called {\em{regular points}}; see Definition~\ref{defregular})
has a smooth manifold structure.
Taking these structures as the starting point, {\em{causal variational principles}} were formulated and studied
as a mathematical generalization of the causal action principle, where an action of the form
\[ \Sact = \int_\F d\rho(x) \int_\F d\rho(y)\: \L(x,y) \]
is minimized for a given lower-semicontinuous Lagrangian~$\L : \F \times \F \rightarrow \R^+_0$
on an (in general non-compact) manifold~$\F$ under variations of~$\rho$ within the class of
regular Borel measures, keeping the total volume~$\rho(\F)$ fixed. We refer the reader interested
in causal variational principles to~\cite[Section~1 and~2]{noncompact} and the references therein.

This article is devoted to the case that the Hilbert space~$\H$ is {\em{infinite-dimensional}}
and separable. While the finite-dimensional setting seems suitable for describing physical spacetime
on a fundamental level (where spacetime can be thought of as being discrete on a microscopic length
scale usually associated to the Planck length), an infinite-dimensional Hilbert space arises in
mathematical extrapolations where spacetime is continuous and has infinite volume.
Most notably, infinite-dimensional Hilbert spaces come up in the examples of causal fermion systems
describing Minkowski space (see~\cite[Section~1.2]{cfs} or~\cite{oppio}) or a globally hyperbolic
Lorentzian manifold (see for example~\cite{nrstg}), and it is also needed for analyzing the limiting
case of a classical interaction (the so-called continuum limit; see~\cite[Section~1.5.2 and Chapters~3-5]{cfs}).
A workaround to avoid infinite-dimensional analysis
is to restrict attention to locally compact variations, as is done in~\cite[Section~2.3]{dirac}.
Nevertheless, in view of the importance of the examples and physical applications, it is a task of
growing significance to analyze causal fermion systems systematically in the infinite-dimensional
setting. It is the objective of this paper to put this analysis on a sound mathematical basis.

We now outline the main points of our constructions and explain our main results.
Extending methods and results in~\cite{gaugefix} to the infinite-dimensional setting, we endow the
set of all regular points of~$\F$ with the structure of a {\em{Banach manifold}}
(see Definition~\ref{defregular} and Theorem~\ref{thmfrechet}). To this end, we construct an
atlas formed of so-called {\em{symmetric wave charts}} (see Definition~\ref{defswc}).
We also show that the Hilbert-Schmidt norm on finite-rank operators on~$\H$ gives rise to
a Fr{\'e}chet-smooth {\em{Riemannian metric}} on this Banach manifold.
More precisely, in Theorems~\ref{thm311} and~\ref{thmriemann}
we prove that~$\F^\reg$ is a smooth Banach submanifold of the
Hilbert space~$\Shil(\H)$ of selfadjoint Hilbert-Schmidt operators, with the Riemannian metric
given by
\[ 
g_x \::\: T_x^\Shil \F^\reg \times T_x^\Shil \F^\reg \rightarrow \R \:,\qquad
g_x(A,B) := \tr(AB) \:. \]

In order to introduce {\em{higher derivatives}} at a regular point~$p \in \F$, our strategy is to
always work in the distinguished symmetric wave chart around this point.
This has the advantage that we can avoid the analysis of differentiability properties under
coordinate transformations. The remaining difficulty is that the causal Lagrangian~$\L$ and other
derived functions are not differentiable. Instead, directional derivatives exist only in certain directions.
In general, these directions do not form a vector space. As a consequence, the derivative
is not a linear mapping, and the usual product and
chain rules cease to hold. On the other hand, these computation rules
are needed in the applications, and it is often sensible to assume that they do hold.
This motivates our strategy of looking for a {\em{vector space}} on which the function under consideration
is differentiable. Clearly, in this way we lose information on the differentiability in certain directions
which do not lie in such a vector space. But this shortcoming is outweighted by the
benefit that we can avoid the subtleties of non-smooth analysis, which, at least for most applications
in mind, would be impractical and inappropriately technical.
Clearly, we want the subspace to be as large as possible, and moreover it should be defined
canonically without making any arbitrary choices. These requirements lead us to the
notion of {\em{expedient subspaces}} (see Definition~\ref{defexpedient}).
In general, the expedient subspace is neither dense nor closed.
On these expedient subspaces, the function is G{\^a}teaux differentiable, the derivative is a linear
mapping, and higher derivatives are multilinear.

The differential calculus on expedient subspaces is compatible with the chain rule in the following sense:
If~$f$ is locally H\"older continuous, $\gamma$ is a smooth curve whose derivatives up to sufficiently high
order lie in the expedient differentiable subspace of~$f$, then the composition~$f \circ \gamma$
is differentiable and the chain rule holds (see Proposition~\ref{prpchain}), i.e.\
\[ (f\circ \gamma)'(t_0) = D^{\E}f|_{x_0}\, \gamma'(t_0) \:, \]
where the index~$\E$ denotes the derivative on the expedient subspace.
We also prove a chain rule for higher derivatives (see Proposition~\ref{prpchain2}).
The requirement of H\"older continuity is a crucial assumption needed in order to control the
error term of the linearization. The most general statement is Theorem~\ref{thmchaingen} where H\"older continuity is
required only on a subspace which contains the curve~$\gamma$ locally.

We also work out how the differential calculus on expedient subspaces applies to the setting of causal
fermion systems. In order to establish the chain rule, we prove that the causal Lagrangian is indeed
locally H\"older continuous with uniform H\"older exponent (Theorem~\ref{thmhoelder}),
and we analyze how the H\"older constant depends on the base point (Theorem~\ref{thmhoelderglobal}).
Moreover, we prove that for all $x,y \in \F$ there is a neighborhood $U\subseteq \F$ of~$y$ with
(see~\eqref{GlobalHoelder2})
\[ | \L(x,y) - \L(x,\tilde{y}) | \leq c(n, y) \|x\|^2\, \| \tilde{y}-y\|^{\frac{1}{2n-1}}
	\qquad \text{for all~$\tilde{y}\in U$} \]
(where~$2n$ is the maximal rank of the operators in~$\F$).
Relying on these results, we can generalize the jet formalism as introduced in~\cite{jet}
for causal variational principles to the infinite-dimensional setting (Section~\ref{secjet}).
We also work out the chain rule for the Lagrangian (Theorem~\ref{thmLchain})
and for the function~$\ell$ obtained by
integrating one of the arguments of the Lagrangian (Theorem~\ref{thmlchain}),
\beq \label{elldef}
\ell(x) = \int_M \L(x,y)\: d\rho(y) - \s
\eeq
(where~$\s$ is a positive constant).

The paper is organized as follows. Section~\ref{secprelim} provides the necessary preliminaries
on causal fermion systems and infinite-dimensional analysis.
In Section~\ref{secFreg} an atlas of symmetric wave charts is constructed, and it is shown
that this atlas endows the regular points of~$\F$ with the structure of a Fr{\'e}chet-smooth Banach manifold.
Moreover, it is shown that the Hilbert-Schmidt norm induces a Fr{\'e}chet-smooth Riemannian metric.
In Section~\ref{secexpedient} the differential calculus on expedient subspaces is developed.
In Section~\ref{secnosmooth}, this differential calculus is applied to causal fermion systems.
Appendix~\ref{appfrechet} gives some more background information on the Fr{\'e}chet derivative.
Finally, Appendix~\ref{appsymm} provides details on how the Riemannian metric looks like in
different charts.

We finally point out that, in order to address a coherent readership,
concrete applications of our methods and results for example to physical spacetimes have
not been included here. The example of causal fermion systems in Minkowski space
will be worked out separately in~\cite{lagrange-hoelder}.

\section{Preliminaries} \label{secprelim}
\subsection{Causal Fermion Systems and the Causal Action Principle} \label{seccfs}
We now recall the basic definitions of a causal fermion system and the causal action principle.

\begin{Def} \label{defcfs} (causal fermion system) {\em{ 
Given a separable complex Hilbert space~$\H$ with scalar product~$\la .|. \ra_\H$
and a parameter~$n \in \N$ (the {\em{``spin dimension''}}), we let~$\F \subseteq \Lin(\H)$ be the set of all
selfadjoint operators on~$\H$ of finite rank, which (counting multiplicities) have
at most~$n$ positive and at most~$n$ negative eigenvalues. On~$\F$ we are given
a positive measure~$\rho$ (defined on a $\sigma$-algebra of subsets of~$\F$), the so-called
{\em{universal measure}}. We refer to~$(\H, \F, \rho)$ as a {\em{causal fermion system}}.
}}
\end{Def} \noindent
A causal fermion system describes a spacetime together
with all structures and objects therein.
In order to single out the physically admissible
causal fermion systems, one must formulate physical equations. To this end, we impose that
the universal measure should be a minimizer of the causal action principle,
which we now introduce. 

For any~$x, y \in \F$, the product~$x y$ is an operator of rank at most~$2n$. 
However, in general it is no longer a selfadjoint operator because~$(xy)^* = yx$,
and this is different from~$xy$ unless~$x$ and~$y$ commute.
As a consequence, the eigenvalues of the operator~$xy$ are in general complex.
We denote these eigenvalues counting algebraic multiplicities
by~$\lambda^{xy}_1, \ldots, \lambda^{xy}_{2n} \in \C$
(more specifically,
denoting the rank of~$xy$ by~$k \leq 2n$, we choose~$\lambda^{xy}_1, \ldots, \lambda^{xy}_{k}$ as all
the nonzero eigenvalues and set~$\lambda^{xy}_{k+1}, \ldots, \lambda^{xy}_{2n}=0$).
We introduce the Lagrangian and the causal action by
\begin{align}
\text{\em{Lagrangian:}} && \L(x,y) &= \frac{1}{4n} \sum_{i,j=1}^{2n} \Big( \big|\lambda^{xy}_i \big|
- \big|\lambda^{xy}_j \big| \Big)^2 \label{Lagrange} \\
\text{\em{causal action:}} && \Sact(\rho) &= \iint_{\F \times \F} \L(x,y)\: d\rho(x)\, d\rho(y) \:. \label{Sdef}
\end{align}
The {\em{causal action principle}} is to minimize~$\Sact$ by varying the measure~$\rho$
under the following constraints:
\begin{align}
\text{\em{volume constraint:}} && \rho(\F) = \text{const} \quad\;\; & \label{volconstraint} \\
\text{\em{trace constraint:}} && \int_\F \tr(x)\: d\rho(x) = \text{const}& \label{trconstraint} \\
\text{\em{boundedness constraint:}} && \iint_{\F \times \F} 
|xy|^2
\: d\rho(x)\, d\rho(y) &\leq C \:, \label{Tdef}
\end{align}
where~$C$ is a given parameter, $\tr$ denotes the trace of a linear operator on~$\H$, and
the absolute value of~$xy$ is the so-called spectral weight,
\[ |xy| := \sum_{j=1}^{2n} \big|\lambda^{xy}_j \big| \:. \]
This variational principle is mathematically well-posed if~$\H$ is finite-dimensional.
For the existence theory and the analysis of general properties of minimizing measures
we refer to~\cite{discrete, continuum, lagrange}.
In the existence theory, one varies in the class of regular Borel measures
(with respect to the topology on~$\Lin(\H)$ induced by the operator norm),
and the minimizing measure is again in this class. With this in mind, here we always assume that
\[ 
\text{$\rho$ is a regular Borel measure}\:. \]

Let~$\rho$ be a {\em{minimizing}} measure. {\em{Spacetime}}
is defined as the support of this measure,
\[ 
M := \supp \rho \:. \]
Thus the spacetime points are selfadjoint linear operators on~$\H$.
These operators contain a lot of additional information which, if interpreted correctly,
gives rise to spacetime structures like causal and metric structures, spinors
and interacting fields. We refer the interested reader to~\cite[Chapter~1]{cfs}.

The only results on the structure of minimizing measures
which will be needed here concern the treatment of the
trace constraint and the boundedness constraint.
As a consequence of the trace constraint, for any minimizing measure~$\rho$
the local trace is constant in spacetime, i.e.\
there is a real constant~$c \neq 0$ such that (see~\cite[Proposition~1.4.1]{cfs})
\[ 
\tr x = c \qquad \text{for all~$x \in M$} \:. \]
Restricting attention to operators with fixed trace, the trace constraint~\eqref{trconstraint}
is equivalent to the volume constraint~\eqref{volconstraint} and may be disregarded.
The boundedness constraint, on the other hand, can be treated with a Lagrange multiplier.
Indeed, as is made precise in~\cite[Theorem~1.3]{lagrange}, for every minimizing measure~$\rho$, 
there is a Lagrange multiplier~$\kappa>0$ such that~$\rho$ is a local minimizer
of the causal action
with the Lagrangian replaced by
\[ \L_\kappa(x,y) := \L(x,y) + \kappa\, |xy|^2 \:, \]
leaving out the boundedness constraint. 

\subsection{Fr{\'e}chet and G{\^a}teaux Derivatives}
We now recall a few basic concepts from the differential calculus on normed vector spaces.
In what follows, we let~$(E, \|.\|_E)$ and~$(F, \|.\|_F)$ be real normed vector spaces.
The most common concept is that of the Fr{\'e}chet derivative.
\begin{Def} \label{deffrechet}
Let~$U \subseteq E$ be open and~$f : U \rightarrow F$ be an $F$-valued function on~$U$.
The function~$f$ is {\bf{Fr{\'e}chet-differentiable}} in~$x_0 \in U$ if there is a
bounded linear mapping~$A \in \Lin(E, F)$ such that
\[ f(x) = f(x_0) + A\, (x-x_0) + r(x) \:, \]
where the error term~$r : U \rightarrow F$ goes to zero faster than linearly, i.e.
\[ \lim_{x \rightarrow x_0, x \neq x_0} \frac{\|r(x)\|_F}{\|x-x_0\|_E} = 0 \:. \]
The linear operator~$A$ is the {\bf{Fr{\'e}chet derivative}}, also denoted by~$Df|_{x_0}$.
A function is Fr{\'e}chet-differentiable in~$U$ if it is Fr{\'e}chet-differentiable at every point of~$U$.
\end{Def} \noindent
The Fr{\'e}chet derivative is uniquely defined. Moreover, the concept can be iterated to
define higher derivatives. Indeed, if~$f$ is differentiable in~$U$, its derivative~$Df$ is a mapping
\[ Df \::\: U \rightarrow \Lin(E,F) \:. \]
Since~$\Lin(E,F)$ is a normed vector space (with the operator norm), we can apply
Definition~\ref{deffrechet} once again to define the second derivative at a point~$x_0$ by
\[ D^2f|_{x_0} = D\big( Df \big) \big|_{x_0} \;\in\; \Lin\big( E, \Lin(E,F) \big) \:. \]
The second derivative can also be viewed as a bilinear mapping from~$E$ to~$F$,
\[ D^2f|_{x_0} : E \times E \rightarrow F\:,\qquad D^2f|_{x_0}(u,v) := \Big( D\big( Df \big) \big|_{x_0} u, v \Big) \:. \]
It is by definition bounded, meaning that there is a constant~$c>0$ such that
\[  \big\| D^2f|_{x_0}(u,v) \big\|_F \leq c \, \|u\|_E\, \|v\|_E \qquad \text{for all~$u, v \in E$}\:. \]
By iteration, one obtains similarly the Fr{\'e}chet derivatives of order~$p \in \N$
as multilinear operators
\[ D^pf|_{x_0} : \underbrace{E \times \cdots \times E}_{\text{$p$ factors}} \rightarrow F \:. \]
A function is {\em{Fr{\'e}chet-smooth}} on~$U$ if it is Fr{\'e}chet-differentiable
to every order.

\begin{Lemma} \label{SymmF-deriv} If the function~$f : U \subseteq E \rightarrow F$ is $p$ times Fr{\'e}chet-differentiable
in~$x_0 \in U$, then its $p^\text{th}$ Fr{\'e}chet derivative is symmetric, i.e.\
for any~$u_1, \ldots, u_p \in E$ and any permutation~$\sigma \in {\mathcal{S}}_p$,
\[ D^pf|_{x_0}\big(u_1, \ldots, u_p \big) = D^pf|_{x_0}\big(u_{\sigma(1)}, \ldots, u_{\sigma(p)} \big) \:. \]
\end{Lemma} \noindent
We omit the proof, which can be found for example in~\cite[Section~4.4]{coleman}.
For the Fr{\'e}chet derivative, most concepts familiar from the finite-dimensional setting
carry over immediately. In particular, the composition of Fr{\'e}chet-differentiable functions
is again Fr{\'e}chet-differentiable. Moreover, the chain and product rules hold.
We refer for the details to~\cite[Sections~2.2~and~2.3]{coleman} and \cite[Chapter~8]{dieudonne1}\footnote{In this reference, everything is worked out in the case of Banach spaces, but the completeness is not needed for these results.} and Appendix~\ref{appfrechet}.

A weaker concept of differentiability which we will use here is
G{\^a}teaux differentiability.
\begin{Def} \label{defgateaux}
Let~$U \subseteq E$ be open and~$f : U \rightarrow F$ be an $F$-valued function on~$U$.
The function~$f$ is {\bf{G{\^a}teaux differentiable}} in~$x_0 \in U$ in the direction~$u \in E$
if the limit of the difference quotient exists,
\[ d_u f(x_0) := \lim_{h \rightarrow 0, h \neq 0} \frac{f(x_0+h u) - f(x_0)}{h} \:. \]
The resulting vector~$d_u f(x_0) \in F$ is the {\bf{G{\^a}teaux derivative}}.
\end{Def} \noindent
By definition, the G{\^a}teaux derivative is homogeneous of degree one, i.e.
\[ d_{\lambda u} f(x_0) = \lambda\, d_u f(x_0) \qquad \text{for all~$\lambda \in \R$}\:. \]
Moreover, if~$f$ is Fr{\'e}chet-differentiable in~$x_0$, then it is also G{\^a}teaux differentiable
in any direction~$u \in E$ and
\[ d_u f(x_0) = Df|_{x_0} u \:. \]
However, the converse is not true because, even if the G{\^a}teaux derivatives
exist for any~$u \in E$, it is in general not possible to represent them by a
bounded linear operator. As a consequence, 
the chain and product rules in general do not hold for G{\^a}teaux derivatives.
We shall come back to this issue in Section~\ref{secnosmooth}.

\subsection{Banach Manifolds} \label{secbanachmanifold}
We recall the basic definition of a smooth Banach manifold
(for more details see for example~\cite[Chapter~73]{zeidlerFA4}).
\begin{Def} \label{defbanachmanifold}
Let~$B$ be a Hausdorff topological space and~$(E, \|.\|_E)$ a Banach space.
A {\bf{chart}}~$(U, \phi)$ is a pair consisting of an open subset~$U \subseteq B$
and a homeomorphism~$\phi$ of~$U$ to an open subset~$V := \phi(U)$ of~$E$, i.e.
\[ \phi \::\: U \overset{\text{\tiny{open}}}{\subseteq} B \rightarrow V \overset{\text{\tiny{open}}}{\subseteq} E \:. \]
A {\bf{smooth atlas}}~${\mathcal{A}} = \big( \phi_i, U_i, E)_{i \in I}$ is a collection of charts
(for a general index set~$I$) with the properties that the domains of the charts cover~$B$,
\[ B = \bigcup_{i \in I} U_i \]
and that for any~$i, j \in I$, the transition map
\[ \phi_j \circ \phi_i^{-1} \::\: \phi_i \big( U_i \cap U_j \big) \subseteq E \rightarrow \phi_j \big( U_i \cap U_j \big) \]
is Fr{\'e}chet-smooth. Two atlases $\big( \phi_i, U_i, E)_{i \in I}$ and $\big( \psi_i, V_i, E)_{j \in J}$ are called {\bf{equivalent}} if all the transition maps $\psi_j \circ \phi_i^{-1}$ and $\phi_i \circ \psi_j^{-1}$ are Fr{\'e}chet-smooth. We denote the corresponding equivalence class by $[\mathcal{A}]$. The union of the charts of all atlases in $[\mathcal{A}]$ is called {\bf{maximal atlas}} $\mathcal{A}_{\mathrm{max}}$.
The triple~$(B, E, {\mathcal{A}})$ is referred to as a {\bf{smooth Banach manifold}} with differentiable structure provided by $\mathcal{A}_{\mathrm{max}}$.
\end{Def}
\begin{Def}
	Just as in the case of finite-dimensional manifolds, we call a function $f: U\subseteq A \rightarrow B$ between two smooth Banach manifolds $(A, E, \mathcal{A})$ and $(B, G, \mathcal{B})$ (with $U\subseteq A$ open) $n$-times \textbf{(Fr{\'e}chet) differentiable} (resp. \textbf{smooth}) if for all combinations of charts $\phi_a:U_a \rightarrow V_a$ and $\phi_b: U_b \rightarrow V_b$ of some (and thus all) atlases $\tilde{\mathcal{A}}$ in $[\mathcal{A}]$ respectively $\tilde{\mathcal{B}}$ in $[\mathcal{B}]$,
	the mapping $\phi_b \circ f\circ \phi_a^{-1}: V_a\rightarrow V_b$ is $n$-times (Fr{\'e}chet) differentiable (resp. smooth).
\end{Def}
\section{Smooth Banach Manifold Structure of~$\F^\reg$} \label{secFreg}
In the definition of causal fermion systems, the number of positive or negative eigenvalues
of the operators in~$\F$ can be strictly smaller than~$n$.
This is important because it makes~$\F$ a closed subspace of~$\Lin(\H)$ (with respect to the
norm topology), which in turn is crucial for the general existence results for minimizers
of the causal action principle (see~\cite{continuum} or~\cite{intro}).
However, in most physical examples in Minkowski space or in a Lorentzian spacetime,
all the operators in~$M$ do have exactly~$n$ positive and exactly~$n$ negative eigenvalues.
This motivates the following definition (see also~\cite[Definition~1.1.5]{cfs}).
\begin{Def} \label{defregular} {\em{
An operator~$x \in \F$ is said to be {\em{regular}} if it has the maximal possible rank,
i.e.~$\dim x(\H) = 2n$. Otherwise, the operator is called {\em{singular}}.
A causal fermion system is {\em{regular}} if all its spacetime points are regular.}}
\end{Def} \noindent
In what follows, we restrict attention to regular causal fermion systems. Moreover, it is
convenient to also restrict attention to all those operators in~$\F$ which are regular,
\[ 
\F^\reg := \big\{ x \in \F \:|\: \text{$x$ is regular} \big\} \:. \]
$\F^\reg$ is a dense open subset of~$\F$ (again with respect to the
norm topology on~$\Lin(\H)$).

\subsection{Wave Charts and Symmetric Wave Charts}
We now choose specific charts and prove that the resulting atlas
endows~$\F^\reg$ with the structure of a smooth Banach manifold
(see Definition~\ref{defbanachmanifold}).
In the finite-dimensional setting, these charts were introduced in~\cite{gaugefix}.
We now recall their definition and generalize the constructions to the
infinite-dimensional setting.

Given~$x \in \F^\reg$ we denote the image of~$x$ by~$I:=x(\H)$.
We consider~$I$ as a $2n$-dimensional Hilbert space with the scalar
product induced from~$\la .|. \ra_\H$. Denoting its
orthogonal complement by~$J:=I^\perp$, we obtain
the orthogonal sum decomposition
\[ \H = I \oplus J \:. \]
This also gives rise to a corresponding decomposition of operators, like for example
\beq \label{lindecomp}
\Lin(\H, I) = \Lin(I,I) \oplus \Lin(J,I)\:.
\eeq
Given an operator~$\psi \in \Lin(\H, I)$, we denote its adjoint by~$\psi^\dagger
\in \Lin(I, \H)$; it is defined by the relation
\[ \la u \,|\, \psi \,v \ra_I = \la \psi^\dagger \,u \,|\, v \ra_\H \qquad \text{for all~$u \in I$ and~$v \in \H$}\:. \]
We now form the operator
\beq \label{Rpsi}
R_x(\psi) := \psi^\dagger \,x\, \psi \in \Lin(\H) \:.
\eeq
By construction, this operator is symmetric and has at most~$n$ positive and at most~$n$
negative eigenvalues. Therefore, it is an operator in~$\F$. 
Using~\eqref{lindecomp}, we conclude that~$R_x$ is a mapping
\beq \label{Rdef1}
R_x \::\: \Lin(I,I) \oplus \Lin(J,I) \rightarrow \F \:.
\eeq

Before going on, it is useful to rewrite the operator~$R_x(\psi)$ in a slightly different way.
On~$I$, one can also introduce the indefinite inner product
\beq \label{ssp}
\Sl .|. \Sr_x \::\: S_x \times S_x \rightarrow \C \:, \qquad 
\Sl u | v \Sr_x = -\la u | x v \ra_\H \:,
\eeq
referred to as the {\em{spin inner product}}. 
For conceptual clarity, we denote~$I$ endowed with the spin inner product
by~$(S_x, \Sl .|. \Sr_x)$ and refer to it as the {\em{spin space}} at~$x$
(for more details on the spin spaces we refer for example to~\cite[Section~1.1]{cfs}).
It is an indefinite inner product space of signature~$(n,n)$.
We denote the adjoint with respect to the spin inner product by a star.
More specifically, for a linear operator~$A \in \Lin(S_x)$, the adjoint is defined by
\[ \Sl  \phi \,|\, A\, \tilde{\phi} \Sr_x = \Sl A^* \,\phi \,|\, \tilde{\phi} \Sr_x \qquad
\text{for all~$\phi, \tilde{\phi} \in S_x$}\:. \]
Using again the definition of the spin inner product~\eqref{ssp}, we can rewrite this equation as
\[ -\la \phi \,|\, X\,A \tilde{\phi} \ra_\H = -\la A^* \phi \,|\,X \tilde{\phi} \ra_\H \:, \]
where we introduced the short notation
\beq \label{Xdef}
X := x|_{S_x} \::\: S_x \rightarrow S_x \:.
\eeq
Taking adjoints in the Hilbert space~$\H$ gives
\[ -\la X^{-1}\, A^\dagger\,X \phi \,|\,X \tilde{\phi} \ra_\H = -\la A^* \phi \,|\,X \tilde{\phi} \ra_\H \]
(note that the operator~$X$ is invertible because~$S_x$ is
by definition its image). We thus obtain the relation
\beq \label{Astar}
A^* = X^{-1}\, A^\dagger\,X \:.
\eeq
Using such transformations, one readily verifies that, identifying the image of~$\psi$
with a subspace of~$S_x$, the right side of~\eqref{Rpsi} can be written as~$-\psi^* \psi$
(for details see~\cite[Lemma~2.2]{gaugefix}).
Thus, with this identification, the operator~$R_x$ can be written instead of~\eqref{Rpsi} and~\eqref{Rdef1}
in the equivalent form
\beq \label{Rdef2}
R_x \::\: \Lin(I,S_x) \oplus \Lin(J,S_x) \rightarrow \F \:,\qquad
R_x(\psi) = -\psi^* \psi \:,
\eeq
where~$\psi^*$ is the adjoint with respect to the corresponding inner products, i.e.
\[ \Sl \phi \,|\, \psi\,u \Sr_x = \la \psi^* \phi \,|\, u \ra_\H \qquad \text{for~$u \in H$ and~$\phi \in S_x$}\:. \]

We want to use the operator~$R_x$ in order to construct local parametrizations of~$\F^\reg$.
The main difficulty is that the operator~$R_x$ is not injective.
For an explanation of this point in the context of local gauge freedom we refer to~\cite{gaugefix}.
Here we merely explain how to arrange that~$R_x$ becomes injective.
We let~$\Symm(S_x) \subseteq \Lin(S_x)$
be the real vector space of all operators~$A$ on~$S_x$ which are
symmetric with respect to the spin inner product, i.e.\
\[ \Sl  \phi | A \tilde{\phi} \Sr_x = \Sl A \phi | \tilde{\phi} \Sr_x \qquad
\text{for all~$\phi, \tilde{\phi} \in S_x$}\:. \]
We now restrict the operator~$R_x$ in~\eqref{Rdef1} and~\eqref{Rdef2} to
\beq \label{Rsymmdef2}
R_x^\symm := R_x|_{\text{Symm}(S_x) \oplus \Lin(J,S_x)} \::\: \Symm(S_x) \oplus \Lin(J,S_x) \rightarrow \F \:,\qquad R_x(\psi) = -\psi^* \psi \:.
\eeq
We write the direct sum decomposition as
\[ \psi = \psi_I + \psi_J \qquad \text{with} \qquad \psi_I \in \text{Symm}(S_x),\; \psi_J \in \Lin(J,S_x)\:. \]

Extending the analysis in~\cite[Section~6.1]{gaugefix} to the infinite-dimensional setting,
one finds that this mapping is a local parametrization of~$\F^\reg$:
\begin{Thm} \label{thmchart} There is an open neighborhood~$W_x$ of~$(\id_{S_x}, 0) \in \mathrm{Symm}(S_x)\oplus \Lin(J,S_x)$
such that the restriction of~$R_x^\symm$ maps to an open subset~$\Omega_x :=  R_x^\symm(W_x)$ of~$\F^\reg$,
\[ R_x^\symm|_{W_x} \::\: W_x \rightarrow \Omega_x
\overset{\text{\tiny{\rm{open}}}}{\subseteq} \F^\reg \:, \]
and is a homeomorphism to its image (always with respect to the topology induced by the operator norm on~$\Lin(\H)$).
\end{Thm}
\Proof 
The estimate
\begin{align}
&\| R_x^\symm(\psi) - R_x^\symm(\tilde{\psi}) \|_{\Lin(\H)} \notag \\
&= \big\| \psi^* \psi - \tilde{\psi}^* \tilde{\psi} \big\|_{\Lin(\H)} \leq
\big\| \psi^* \psi - \psi^* \tilde{\psi} \big\|_{\Lin(\H)} + \big\| \psi^* \tilde{\psi} - \tilde{\psi}^* \tilde{\psi} \big\|_{\Lin(\H)} \notag \\
&\leq \| \psi^*\|_{\Lin(\H)} \, \big\| \psi - \tilde{\psi} \big\|_{\Lin(\H)} + \big\| \tilde{\psi}^* - \tilde{\psi}^* \big\|_{\Lin(\H)}
\: \| \tilde{\psi}\|_{\Lin(\H)}  \label{simplees}
\end{align}
shows that~$R_x^\symm$ is continuous. Since the point~$R_x^\symm(\id_{S_x}, 0)=x \in \F^\reg$
is regular, by continuity we may choose an open neighborhood~$W_x$ of~$(\id_{S_x}, 0)$
such that~$R_x$ maps to~$\F^\reg$.

In order to show that~$R_x^\symm$ is bijective, we 
begin with the formula for~$\phi_x$ as derived in~\cite[Proposition~6.6]{gaugefix},
which will turn out to be the inverse of~$R_x^\symm$. It has the form
\beq \label{phixplicit}
\phi_x(y) = \big( P(x,x)^{-1}\, A_{xy}\, P(x,x)^{-1} \big)^{-\frac{1}{2}}\, P(x,x)^{-1}\, P(x,y) \:
\Psi(y) \;\in\; \Lin(\H, S_x) \:,
\eeq
where~$P(x,y)$ (the {\em{kernel of the fermionic projector}}) and~$A_{xy}$ (the {\em{closed chain}}) are defined by
\beq \label{Pxydef}
P(x,y) := \pi_x y|_{S_y} \::\: S_y \rightarrow S_x \:,\qquad A_{xy} := P(x,y)\, P(y,x) \::\: S_x \rightarrow S_x \:.
\eeq
Our task is to show that for a sufficiently small open neighborhood~$\Omega_x$ of~$x$,
this formula defines a continuous mapping
\[ 
\phi_x \::\: \Omega_x \subseteq \F^\reg \rightarrow \Symm(S_x) \oplus \Lin(J, \H) \:, \]
and that the compositions
\beq \label{compose}
\phi_x \circ R_x^\symm|_{W_x} \qquad \text{and} \qquad R_x^\symm \circ \phi_x
\eeq
are both the identity (showing that~$\phi_x$ is indeed the inverse of~$R_x^\symm$).

In preparation, we rewrite the formula~\eqref{phixplicit} as
\beq
\phi_x(y) = \Big(X^{-1} \,\pi_x y\pi_y x X^{-1}\Big)^{-\frac{1}{2}}X^{-1} \,\pi_x y \pi_y = \Big(X^{-1} \,
\pi_x y |_{S_x}\Big)^{-\frac{1}{2}}X^{-1}\,\pi_x y \:, \label{phixy}
\eeq
where we again used the notation~\eqref{Xdef}.
Choosing~$y=x$, the operator~$X^{-1} \,\pi_x y |_{S_x}$ is the identity on~$S_x$.
We first choose an open neighborhood~$\tilde{\Omega}_x$ of $x$ so small such that for any $y \in \tilde{\Omega}_x$,
\beq \label{invertible}
\big\| \id_{S_x}-X^{-1}\pi_x y |_{S_x} \big\|_{\Lin(\mathcal{H})}< \frac{1}{2} \:.
\eeq
Then the square root as well as the inverse square root of $A=X^{-1}\pi_x y$  are well-defined for all $x\in \tilde{\Omega}_x$ by the respective power series,
\begin{flalign*}
A^{\frac{1}{2}}:=\sum_{n=0}^{\infty} (-1)^n \binom{1/2}{n}(\id_{S_x}-A)^n\;,\;\;\;
A^{-\frac{1}{2}}:=\sum_{n=0}^{\infty} (-1)^n \binom{-1/2}{n}(\id_{S_x}-A)^n\;,
\end{flalign*}
with the generalized binomial coefficients given for $\beta \in \mathds{R}$ and $n\in \mathds{N}$ by
\begin{flalign*}
\binom{\beta}{n}:= \left\{ \begin{array}{cl}
\displaystyle \frac{1}{n!}\;\beta \cdot (\beta - 1)\cdots (\beta-n+1)\quad & \text{if~$n>0$} \\
0\quad&\text{if~$n=0$}
\end{array} \right.
\end{flalign*}
as for both power series the radius of convergence equals one. Moreover note that all square roots, inverse square roots, etc.\ appearing in the following are well-defined as they are always applied to operators within their radius of convergence.
We conclude that the mapping~$\phi_x$ is well-defined and continuous on~$\tilde{\Omega}_x$. Now by possibly shrinking $W_x$ we can arrange that $\Omega_x:=R_x^\symm(W_x)$ lies in $\tilde{\Omega}_x$. Note that it now suffices to show that $\phi_x|_{\Omega_x}$ is the inverse of $R_x^\symm|_{W_x}$, because then
the set~$\Omega_x=(\phi_x|_{\tilde{\Omega}_x})^{-1}(W_x)$ is open.

In order to verify that~$\phi_x$ maps into $\mathrm{Symm}(S_x)\oplus \Lin(J,S_x)$,
we restrict~$\phi_x(y)$ to~$S_x$,
\begin{flalign}
\phi_x(y) \big|_I&=\Big( \Big(X^{-1} \,\pi_x \,y \big|_{S_x}\Big)^{-\frac{1}{2}}X^{-1} \,\pi_x \,y \Big)\Big|_I \notag \\
	&= \Big(X^{-1} \,\pi_x \,y \big|_{S_x}\Big)^{-1/2} X^{-1} \,\pi_x \,y|_{S_x} = \Big(X^{-1} \,\pi_x \,y \,\pi_x \big|_{S_x}\Big)^{\frac{1}{2}} \:. \label{phixysymm}
\end{flalign}
A direct computation using~\eqref{Astar} shows that the operator~$X^{-1}\pi_x y \pi_x|_{S_x}$,
and hence also its square root, are symmetric on~$S_x$.

It remains to compute the compositions in~\eqref{compose}. First,
\begin{align*}
\phi_x \circ R_x^\text{\tiny{symm}}(\psi)
&= \phi_x(\psi^{\dagger}X\psi) 
=\Big(X^{-1}\underbrace{\pi_x \psi^{\dagger}X}_{=\psi_I^{\dagger}X} \underbrace{\psi |_{S_x}}_{\psi_I}\Big)^{-\frac{1}{2}}X^{-1}\underbrace{\pi_x \psi^{\dagger}X}_{=\psi_I^{\dagger}X}\psi \\
&= \Big( \underbrace{ X^{-1}\psi_I^{\dagger}X}_{=\psi_I}\psi_I\Big)^{-\frac{1}{2}}\underbrace{ X^{-1}\psi_I^{\dagger}X}_{=\psi_I}\psi  = \big(\psi_I^2 \big)^{-\frac{1}{2}}\,\psi_I\,\psi =\psi \:,
\end{align*}
where in the last line we applied~\eqref{Astar} and used that~$\psi_I$ is symmetric on~$S_x$.
Moreover,
\begin{align*}
R_x^\symm \circ \phi_x(y) &= \phi_x(y)^\dagger X \phi_x(y) \\
&= y\,\pi_x\, X^{-1}\,\Big(\pi_x y \pi_x\:X^{-1}\: \Big)^{-\frac{1}{2}}\:X\: \Big(X^{-1}\,\pi_x y \pi_x|_{S_x} \Big)^{-\frac{1}{2}} X^{-1}\,\pi_x \,y \:.
\end{align*}
Since the spectral calculus is invariant under similarity transformations, we know that for any
invertible operator~$B$ on~$S_x$,
\[ X^{-1} B^{-\frac{1}{2}} X = \Big( X^{-1} B X \Big)^{-\frac{1}{2}} \:. \]
Hence
\begin{align*}
R_x^\symm \circ \phi_x(y) 
&= y\,\pi_x\, \Big(X^{-1}\,\pi_x y \pi_x|_{S_x}\Big)^{-\frac{1}{2}}\:\Big(X^{-1} \,\pi_x y \pi_x|_{S_x} \Big)^{-\frac{1}{2}} X^{-1}\,\pi_x \,y \\
&= y\,\pi_x\, \Big(X^{-1}\,\pi_x y \pi_x|_{S_x}\Big)^{-1}\:X^{-1}\,\pi_x \,y \\
&= y\,\pi_x\, \Big( \pi_x \,y \pi_x |_{S_x} \Big)^{-1}\, \pi_x \,y
= y\,x \Big( \pi_x \,y  x|_{S_x} \Big)^{-1}\, \pi_x \,y \\
&= y\,P(y,x)\:\Big( P(x,y)\, P(y,x) \Big)^{-1} \, P(x,y) = y
\end{align*}
(note that~$P(x,y) : S_y \rightarrow S_x$ is invertible in view of~\eqref{invertible}).
This concludes the proof.
\QED

The mapping~$\phi_x$, which already appeared in the proof of the previous lemma, can also be introduced
abstractly to define the chart.
\begin{Def} \label{defswc}
Setting
\[ \phi_x := R_x^\symm \big|_{W_x}^{-1} \::\: \Omega_x \rightarrow \Symm(S_x) \oplus \Lin(J,S_x) \:, \]
we obtain a chart~$(\phi_x, \Omega_x)$, referred to as the {\bf{symmetric wave chart}} about the point~$x \in \F^\reg$.
\end{Def} \noindent

We remark that more general charts can be obtained by restricting~$R_x$ to another subspace of~$\Lin(I,S_x) \oplus \Lin(J,S_x)$, i.e.\ in generalization of~\eqref{Rsymmdef2},
\[ R_x^E := R_x|_{E \oplus \Lin(J,S_x)} \::\: E \oplus \Lin(J,S_x) \rightarrow \F \:,\qquad R(\psi) = -\psi^* \psi \:, \]
where~$E$ is a subspace of~$\Lin(S_x)$ which has the same dimension as~$\Symm(S_x)$.
The resulting charts~$\phi^E_x$ are obtained by composition with a unitary operator~$U_x$ on~$S_x$, i.e.
\[ \phi^E_x = U_x \circ \phi_x \qquad \text{with} \qquad U_x \in \U(S_x) \]
(for details and the connection to local gauge transformations see~\cite[Section~6.1]{gaugefix}).
Since linear transformations are irrelevant for the question of differentiability, in what follows we may restrict
attention to symmetric wave charts.

\subsection{A Fr{\'e}chet Smooth Atlas}
The goal of this section is to prove that the symmetric wave charts~$(\phi_x, \Omega_x)$
form a smooth atlas of~$\F^\reg$.

\begin{Thm}[{\bf{Symmetric wave atlas}}] \label{thmfrechet}
The collection of all symmetric wave charts on~$\F^\reg$ defines a Fr{\'e}chet-smooth atlas of
$\F^\reg$, endowing~$\F^\reg$ with the structure of a smooth Banach manifold
(see Definition~\ref{defbanachmanifold}).
\end{Thm}
\begin{proof}
We first verify that for any $x\in \F^\reg$, the vector space~$\mathrm{Symm}(S_x)\oplus \Lin(J,S_x)$ together with the operator norm of $\Lin(\H,I)=\Lin(\H,S_x)$ is a Banach space.
To this end, we note that this vector space coincides with the kernel of the mapping~$\psi \mapsto (X^{-1}\psi^\dagger \pi_x X - \psi|_I)$ on~$\Lin(\H, I)$.
Since this mapping is continuous on~$\Lin(\H,I)$ (as one verifies by an estimate similar to~\eqref{simplees}), its kernel is closed. As a consequence, the vector space~$\mathrm{Symm}(S_x)\oplus \Lin(J,S_x)$ is
a closed subspace of~$\Lin(\H,I)$ and thus indeed a Banach space.

We saw in Theorem~\ref{thmchart} that for any $x \in \F^\reg$, $(\phi_x, \Omega_x)$ defines a 
chart on~$\F^\reg$. Since the~$\Omega_x$ clearly cover $\F^\reg$,
it remains to show that all transition mappings are Fr{\'e}chet-smooth. To this end, we first note that
for any~$x,y \in \F^\reg$ and~$\psi \in \phi_x(\Omega_x \cap \Omega_y)$,
	\begin{flalign*}
	\phi_y \circ \phi_x^{-1}(\psi) = \phi_y \big(\psi^{\dagger} \,X\, \psi \big)
	= \Big(Y^{-1} \,\pi_y\, \psi^{\dagger} \,X\, \psi|_{S_y}\Big)^{-\frac{1}{2}}\:Y^{-1} \,\pi_y\,\psi^{\dagger}\, X\, \psi\;.
	\end{flalign*}
Next, we define the mappings
\begin{align*}
&B_{xy}: \mathrm{Symm}(S_x)\oplus \Lin(J,S_x) \rightarrow \Lin(\H,S_y)\:,&&\psi \mapsto Y^{-1} \,\pi_y
\,\psi^{\dagger} \,X\, \psi\;,\\
&\tilde{B}_{xy}: \mathrm{Symm}(S_x)\oplus \Lin(J,S_x) \rightarrow \Lin(S_y)\:,&& \psi \mapsto
Y^{-1} \,\pi_y \,\psi^{\dagger} \,X\, \psi|_{S_y}\;,\\
&W: B_{\frac{1}{2}}(0)\subseteq \Lin(S_y) \rightarrow \Lin(S_y)\:,&&B\mapsto (1+B)^{-\frac{1}{2}} = \sum_{n=0}^{\infty}(-1)^n \binom{-1/2}{n} \:B^n
\end{align*}
(where the radius of the ball~$B_{1/2}(0)$ is taken with respect to the operator norm).

Recall that in the proof of Theorem~\ref{thmchart} (more precisely \eqref{invertible}) we chose~$\Omega_y$ so small that the
operator~$\| \id_{S_y}- Y^{-1}\pi_yz|_{S_y} \| <1/2$ for any $z\in \Omega_y$. Thus, since for any $\psi \in \phi_x(\Omega_x \cap \Omega_y)$ we have~$\psi^{\dagger}X\psi=\phi^{-1}_x(\psi)\in \Omega_y$, we obtain~$\tilde{B}_{xy}(\phi_x(\Omega_x \cap \Omega_y)) \subseteq B_{1/2}(\id_{S_y})$.
Therefore, we can write the transition mapping~$\phi_y \circ \phi_x^{-1}$ as
	\begin{flalign*}
	\phi_y \circ \phi_x^{-1}(\psi) = W\big(\id_{S_y}-\tilde{B}_{xy}(\psi)\big) \circ B_{xy}(\psi)\;.
	\end{flalign*}
Now note that for the Fr{\'e}chet derivative, we consider all vector spaces here as 
a {\em{real}} Banach spaces, but still with the canonical operator norm induced by $\|.\|_{\H}$.
In view of the chain rule for Fr{\'e}chet derivatives (for details see Lemma~\ref{F-derivative conecttinating}
in Appendix~\ref{appfrechet}) and the properties of the Fr{\'e}chet derivative in Lemma~\ref{Properties F-derivative} in Appendix~\ref{appfrechet}, it remains to show that  the mappings~$W$, $B_{xy}$ and~$\tilde{B}_{xy}$ are Fr{\'e}chet-smooth (note that the composition operator of $\mathds{R}$-linear mappings is also always Fr{\'e}chet-smooth as it defines a bounded $\mathds{R}$-bilinear map and the map $\Lin(S_y)\ni y \mapsto \id_{S_y}-y\in \Lin(S_y)$ is clearly Fr{\'e}chet-smooth as well).
	For $W$ this is clear due to \cite[p. 40--42]{hilgert+neeb} (note that $\Lin(S_y)$ obviously defines a finite-dimensional unital Banach-algebra).
	Moreover, the mappings $B_{xy}$ and $\tilde{B}_{xy}$ are obviously $\mathds{R}$-bilinear and bounded and thus Fr{\'e}chet-smooth.
\end{proof}

\subsection{The Tangent Bundle}
Having endowed~$\F^\reg$ with a canonical smooth Banach manifold structure, the next step is to consider its tangent bundle. 
For finite-di\-men\-sio\-nal manifolds, the tangent space can be defined either by equivalence classes of curves or by derivations, and these two definitions coincide (see for example~\cite[Chapter~2]{lee-riemannian}).
In infinite dimensions, however, this does no longer be the case: In general, the derivation-tangent vectors (usually called \textit{operational tangent vectors}) form a larger class of than the curve-tangent vectors (called \textit{kinematic tangent vectors}). There might even be operational tangent vectors that depend on higher-order derivatives of the inserted function (while the kinematic tangent vectors interpreted as directional derivatives only involve the first derivatives); for details on such issues see for example~\cite[Sections~28 and~29]{kriegl+michor} or \cite[p. 3--6]{belticua}. It turns out that for our applications in mind,
it is preferable to define tangent vectors as equivalence classes of curves.
Indeed, as we shall see, with this definition the usual computation rules
remain valid. More specifically, the tangent vectors of $\F^\reg$ are compatible with the Fr{\'e}chet derivative,
and each fiber of the corresponding tangent bundle can be identified with the underlying Banach space
\[ V_x :=\mathrm{Symm}(S_x)\oplus \Lin(J,S_x) \]
with respect to the chart $\phi_x$.

Following~\cite[p.~284]{kriegl+michor}, we begin with the abstract definition of the (kinematic)
tangent bundle, which makes it easier to see the topological structure. Afterward, we will show that this 
notion indeed agrees with equivalence classes of curves.
Given~$x' \in \F^\reg$, we consider the set~$\Omega_{x'} \times V_{x'} \times \{x'\}$
(endowed with the topology inherited from the direct sum of Banach spaces).
We take the disjoint union
\begin{flalign*}
	\bigcup\limits_{x' \in \F^\reg} \Omega_{x'} \times V_{x'} \times \{x'\}
	\end{flalign*}
and introduce the equivalence relation
	\begin{flalign*}
	(x,\bv,x') \sim (y,\bw,y') \qquad \Longleftrightarrow \qquad
	x=y \quad \mathrm{and} \quad (\phi_{x'}\circ \phi_{y'}^{-1})'|_{\phi_{y'}(x)}\bw = \bv\;.
	\end{flalign*}
	For clarity, we point out that the first entry represents the point of the Banach manifold~$\F^\reg$,
	whereas the third entry labels the chart.
\begin{Def} We define the tangent bundle $T\F^\reg$ as the quotient space with respect to this equivalence relation,
	\begin{flalign*}
	T\F^\reg := \Big(\bigcup\limits_{x' \in \F^\reg} \Omega_{x'} \times V_{x'} \times \{x'\} \Big)\Big/ \sim\;.
	\end{flalign*}
The canonical projection is given by
	\begin{flalign*}
	\pi: T\F^\reg \rightarrow \F^\reg\;,\;\;\;\pi([x,\bv,x']) = x\;.
	\end{flalign*}
	For every~$x \in \F^\reg$ the tangent space at~$x$ is defined by
	\begin{flalign*}
	T_x\F^\reg:=\pi^{-1}(x)\;.
	\end{flalign*}
\end{Def}
Note that each $T_x\F^\reg$ has a canonical vector space structure in the following sense: Since all equivalence classes in $T_x\F^\reg$ have a representative of the form $[x,\bv,x]$, this representative can be identified with $\bv \in V_x$. In this way, we obtain an identification of $T_x\F^\reg$ with $V_x$.

The tangent bundle is again a Banach manifold, as we now explain. For any~$x \in \F^\reg$, the mapping
	\begin{flalign*}
	(\phi_x, D\phi_x): \pi^{-1}(W_x) \rightarrow \Omega_x \times V_x\;,\;\;\;
	[y,\bv,z] \mapsto \Big(\phi_x(y), D\big( \phi_x \circ \phi_z^{-1} \big) \big|_{\phi_z(y)}\bv \Big)
	\end{flalign*}
has the inverse
	\begin{flalign*}
	(\phi_x, D\phi_x)^{-1}: \Omega_x \times V_x \rightarrow \pi^{-1}(W_x)\;,\;\;\;
	(\psi,\bv) \mapsto [\phi_x^{-1}(\psi),\bv,x] \:.
	\end{flalign*}
On~$T\F^\reg$ we choose the coarsest topology with the property that the natural
projections of these mappings to~$\Omega_x$ and~$V_x$ are both continuous
(where on~$\Omega_x$ and~$V_x$ we choose the topology induced by the norm topology of~$\Lin(\H)$).
With this topology, the mapping~$(\phi_x, D\phi_x)$
defines a chart of~$T\F^\reg$.
For any~$(\psi,\bv)\in (\phi_y, D\phi_y)\big(\pi^{-1}(\Omega_x)\cap \pi^{-1}(\Omega_y)\big)$, the transition mappings are given by
	\begin{flalign*}
	(\phi_x, D\phi_x) \circ (\phi_y, D\phi_y)^{-1}(\psi, \bv) &=
	(\phi_x, D\phi_x)([\phi_y^{-1}(\psi), \bv,y])\\
	&= \Big((\phi_x \circ \phi_y^{-1})(\psi), D \big(\phi_x \circ \phi_y^{-1} \big) \big|_{\psi}\bv \Big) \:.
	\end{flalign*}
\begin{Prp} $T\F^\reg$ is again a Banach manifold.
\end{Prp}
\Proof We need to show that transition maps are Fr{\'e}chet-smooth.
This is clear for the first component because the transition mappings $\phi_x \circ \phi_y^{-1}$ are Fr{\'e}chet-smooth
and fiberwise linear. The second component can be considered as the composition of the insertion map
\[ \Lin(V_y,V_x)\times V_y \ni (A,\bv) \mapsto A(\bv)\in V_x \]
(which is obviously continuous and bilinear and thus Fr{\'e}chet-smooth, for details see Lemma~\ref{Properties F-derivative} in Appendix~\ref{appfrechet}) with the mapping $W_y\times V_y \ni (\psi,\bv) \mapsto ((\phi_x \circ \phi_y^{-1})'|_{\psi},\bv)\in \Lin(V_x,V_y)\times V_y$, which is Fr{\'e}chet-smooth due to the Fr{\'e}chet-smoothness of the transition mappings.
\QED
\noindent 
In what follows, we will sometimes use the notation
	\begin{flalign*}
	D\phi_x([y,\bv,z]) := D\big( \phi_x \circ \phi_z^{-1} \big) \big|_{\phi_z(y)} \,\bv \qquad \forall x\in \F^\reg,\; [y,\bv,z]\in \pi^{-1}(\Omega_x)\;,
	\end{flalign*}
which also clarifies the independence of the choice of representatives.
\begin{Lemma}
For any $x \in \F^\reg$, the mapping
	\begin{flalign*}
	\psi_x: \Omega_x \times V_x \rightarrow \pi^{-1}(\Omega_x)\;,\qquad(y, v) \mapsto [y,\bv,x]
	\end{flalign*}
	is a local trivialization.
\end{Lemma}
\Proof We need to verify the properties of a local trivialization. Clearly, the operator~$\pi \circ \psi_x$ is the projection to the first component, and for fixed $y \in \Omega_x$, the mapping $v \mapsto \psi_x(y,\bv)=[y,\bv,x]=[y, (\phi_y \circ \phi_x^{-1})'|_{\phi_x(x)}\bv,y]$ corresponds to $\bv \mapsto (\phi_y \circ \phi_x^{-1})'|_{\phi_x(x)}\bv$ (by the identification of $T_y\F^\reg$ with $V_y$ from before), which is obviously an isomorphism of vector spaces in view of Lemma~\ref{Properties F-derivative}~(vi).
\QED
To summarize, the Banach manifold $\F^\reg$ has similar properties as in the finite-dimen\-sio\-nal case.

We now explain how the above definition of tangent vectors
relates to the equivalence classes of curves (following \cite[p.~285]{kriegl+michor}):
\begin{Remark} {\bf{(equivalence classes of curves)}} {\em{
On curves $\gamma, \tilde{\gamma} \in C^{\infty}(\mathds{R},\F^\reg)$, we consider the equivalence
relation~$\gamma \sim \tilde{\gamma}$ defined by the conditions that~$\gamma(0) = \tilde{\gamma}(0)$
and that in a chart~$\phi_x$ with~$\gamma(0) \in \Omega_x$, the relation~$(\phi_x\circ \gamma)'|_0= (\phi_x\circ \tilde{\gamma})'|_0$ holds.
Note that if the last relation holds in one chart, then
it also holds in any other chart $\phi_y$ with $\gamma(0)\in \Omega_y$ because, due to the chain rule,
	\begin{flalign*}
	(\phi_y\circ \gamma)'|_0 &= (\phi_y \circ \phi_x^{-1} \circ \phi_x \circ \gamma)'|_0 = (\phi_y \circ \phi_x^{-1})'|_{\phi_x	(\gamma(0))}(\phi_x\circ \gamma)'|_0\\ 
	&= (\phi_y \circ \phi_x^{-1})'|_{\phi_x	(\gamma(0))}(\phi_x\circ \tilde{\gamma})'|_0 =(\phi_y \circ \phi_x^{-1} \circ \phi_x \circ \tilde{\gamma})'|_0 = (\phi_y\circ \tilde{\gamma})'|_0\;.
	\end{flalign*}
	Now we can identify $C^{\infty}(\mathds{R},\F^\reg) / \sim$ with $T\F^\reg$ via the mapping
\beq \label{gamma0}
\begin{split}
	C^{\infty}(\mathds{R},\F^\reg) / \sim &\rightarrow T\F^\reg \\
	[\gamma] &\mapsto \Big[\gamma(0)\,,\, (\phi_{\gamma(0)}\circ \gamma)'|_0\,,\, \gamma(0)\Big]\;, 
\end{split}
\eeq
	which bijective with inverse (for details see~\cite[p. 285]{kriegl+michor})
	\begin{flalign*}
	[x,\bv,x'] \mapsto \Big[t\mapsto \phi_{x'}^{-1}\Big(\phi_{x'}(x)+ t \,\xi_{\bv}(t)\,\bv\Big) \Big]\;,
	\end{flalign*}
	where~$\xi_\bv \in C_0^{\infty}(\mathds{R})$ is a smooth cutoff function with $0\leq \xi_v \leq 1$.
	Moreover, $\mathrm{supp}(\xi_\bv)\subseteq (-\varepsilon,\varepsilon)$ and $\xi_\bv|_{(-\varepsilon/2,\varepsilon/2)}\equiv 1$ with $\varepsilon>0$ chosen so small that
\[	B_{\varepsilon\|\bv\|}\big( \phi_{x'}(x) \big) \subseteq W_{x'} \:. \] 

Note that in~\eqref{gamma0} the tangent vector at~$\gamma(0)$ was expressed
in the specific chart $(\phi_{\gamma(0)}, \Omega_{\gamma(0)})$. However, the tangent vector can
also be represented in another chart as follows.
Let~$x \in \F^\reg$ and $[x,\bv,z] \in T_x\F^\reg$ be arbitrary. We say that a curve $\gamma \in C^{\infty}(\mathds{R},\F^\reg)$ represents $[x,\bv,z]$ if  in one chart~$\phi_y$ with $x \in \Omega_y$ (and thus any chart, as one can show using the chain rule just as before) it holds that
	\begin{flalign}
	\label{curve rep. v}
	[x,\bv,z] = [\gamma(0),(\phi_y \circ \gamma)'|_0, y]\;.
	\end{flalign}
	In order to show independence of~$y$, let~$w\in \F^\reg$ with $x \in \Omega_w$. Then
	\begin{flalign*}
	(\phi_w\circ \gamma)'|_0 = (\phi_w \circ \phi_y^{-1} \circ \phi_y \circ \gamma)'|_0 = (\phi_w \circ \phi_y^{-1})'|_{\phi_y(x)}(\phi_y\circ \gamma)'|_0\;,
	\end{flalign*}
	and thus
	\begin{flalign*}
	[\gamma(0), (\phi_w\circ \gamma)'|_0, w] = [\gamma(0),(\phi_y \circ \gamma)'|_0, y] = [x,\bv,z] \;.
	\end{flalign*}
	Hence if~\eqref{curve rep. v} holds in one chart, it also holds in any other chart around $x$.
	}} 
$\text{ }$	\QEDrem 
\end{Remark}
\begin{Remark} {\bf{(directional derivatives)}} {\em{
Let $\gamma \in C^{\infty}(\mathds{R},\F^\reg)$ be a curve that represents $[x,\bv,z]$. We define the directional derivative of a Fr{\'e}chet-differentiable function $f: \F^\reg \rightarrow \mathds{R}$ at $x$ in the direction $[x,\bv,z]$ as
	\begin{flalign*}
	D_{[x,\bv,z]}f|_x := \frac{d}{dt}(f\circ \gamma)|_{t=0}\;.
	\end{flalign*}
	This definition is independent of the choice of the curve $\gamma$. Indeed, for any chart $\phi_w$ around $x$, we have
	\begin{flalign*}
	&\frac{d}{dt}(f\circ \gamma)|_{t=0} = (f\circ \phi_w^{-1} \circ \phi_w \circ \phi_z \circ \phi_z^{-1} \gamma)'(0) \\
	&= D(f\circ \phi_w^{-1})|_{\phi_w(x)}\: D(\phi_w \circ \phi_z^{-1})|_{\phi_z(x)}\:
	(\phi_z \circ \gamma)'(0) \\
	&=D(f\circ \phi_w^{-1})|_{\phi_w(x)}\: D(\phi_w \circ \phi_z^{-1})|_{\phi_z(x)} \:v
	=D(f\circ \phi_w^{-1})|_{\phi_w(x)}\: D\phi_w([x,\bv,z]) \;.
	\end{flalign*}
}} 
\vspace*{-2.7em}

\QEDrem
\end{Remark}
We close this subsection with one last definition:
\begin{Def} {\bf{(Tangent vector fields)}}
	A tangent vector field on a Banach manifold is \hbox{-- similar to the finite-dimensional case --} a Fr{\'e}chet-smooth map $\bv: \F^\reg \rightarrow T\F^\reg$ such that $\bv(x) \in T_x\F^\reg$ (i.e. $\pi(\bv(x)) = x$) for all $x \in \F^\reg$.
	We denote the set of all tangent vectors fields of $\F^\reg$ by~$\Gamma(\F^\reg,T\F^\reg)$.
\end{Def}

We note that, according to this definition, multiplying a vector field by Fr{\'e}chet-smooth real-valued function
gives again a vector field. In other words,
the space of all tangent vector fields forms a module over the ring of Fr{\'e}chet-smooth functions from~$\F^\reg$
to~$\R$.

\subsection{A Riemannian Metric}
\label{SecRiemMetric}
In this section we show that the Hilbert-Schmidt scalar product gives rise to a canonical Riemannian
metric on~$\F^\reg$. For the constructions, it is most convenient to recover~$\F^\reg$ as a
Banach submanifold of the real Hilbert space~$\Shil(\H)$ of all selfadjoint Hilbert-Schmidt operators on~$\H$
endowed with the scalar product ($\Shil$ because of the second Schatten class;
for details see~\cite[Section~XI.6]{dunford2})
\[ \la A, B\ra_{\Shil(\H)} := \tr \big(A B \big) \:. \]

\begin{Thm} \label{thm311}
$\F^\reg$ is a smooth Fr{\'e}chet submanifold of~$\Shil(\H)$ in the following sense.
Given~$x \in \F^\reg$, we choose~$\psi_0 \in \Symm(S_x) \oplus \Lin(J, I)$
with~$x = -\psi_0^* \psi_0$. Then the mapping
\begin{align*}
\mycal{R} \::\: \big( \Symm(S_x) \oplus \Lin(J, I) \big) \oplus \Shil(J) &\rightarrow \Shil(\H) \\
(\;\;\psi\;\;,\;\;  B\;\;)\: &\mapsto -\psi^* \psi + \begin{pmatrix} 0 & 0 \\ 0 & B \end{pmatrix}
\end{align*}
(where the last matrix denotes a block operator on~$\H=I \oplus J$)
is a local Fr{\'e}chet-diffeomorphism at~$(\psi_0, 0)$. Its local inverse takes the form
\begin{align*}
\Phi := ({\mycal{R}}|_{\hat{W}})^{-1} \::\: \Shil(\H) \cap \hat{\Omega}_x &\rightarrow \hat{W} \subseteq 
\big( \Symm(S_x) \oplus \Lin(J, I) \big) \oplus \Shil(J) \\
E &\mapsto \bigg( \phi_x(\pi_x E), \pi_J \Big( E + \big(\phi_x(\pi_x E) \big)^* \phi_x(\pi_x E) \Big) \Big|_J \bigg) \:,
\end{align*}
where~$\hat{W} = W_x \oplus \Shil(J)$, $\hat{\Omega}_x:=\mycal{R}(\hat{W})=\Omega_x+\Shil(J)$ (with~$W_x$ and~$\Omega_x$ as in Theorem~\ref{thmchart}), and~$\phi_x(\pi_x E)$ is
defined in analogy to~\eqref{phixy} by
\[ \phi_x \big(\pi_x E \big) := \Big(X^{-1} \,\pi_x E |_{S_x}\Big)^{-\frac{1}{2}}X^{-1}\,\pi_x E \:\in\: \Symm(S_x) \oplus \Lin(J, I) \]
(the fact that this maps to the symmetric operators on~$S_x$ is verified as in~\eqref{phixysymm}).
\end{Thm}
\Proof A direct computation shows that~${\mycal{R}}$ and~$\Phi$ are inverses of each other:
In order to compute~$\mycal{R}\circ \Phi$, we use the block operator notation
\[E= \begin{pmatrix} E_{II} & E_{IJ} \\ E_{JI} & E_{JJ} \end{pmatrix}  \in \Shil(\H) \cap \hat{\Omega}_x \:. \]
Then there exist operators $\tilde{E}_J, \hat{E}_J \in \Shil(J)$ such that $E_{JJ}=\tilde{E}_J+\hat{E}_J$, and the operator
\[ \tilde{E}:= \begin{pmatrix} E_{II} & E_{IJ} \\ E_{JI} & \tilde{E}_J \end{pmatrix} \]
is contained in $\Omega_x$. Note that $\phi_x E = \pi_x\tilde{E}$ and therefore $-\phi_x(\pi_x E)^*\phi_x(\pi_x E)=\tilde{E}$. We conclude that
\[\mycal{R}\circ \Phi (E) = \begin{pmatrix} E_{II} & E_{IJ} \\ E_{JI} & \tilde{E}_J+E_{JJ}-\tilde{E}_J \end{pmatrix} =E \:. \]
In order to compute~$\Phi \circ \mycal{R}$, we take $(\psi, B)\in \hat{W}$ arbitrary and note that, due to the definition of $\phi_x$ in \eqref{phixy} and Theorem~\ref{thmchart}, we have
\[ \phi_x(\mycal{R}(\psi,B)) = \phi_x(-\pi_x\psi^*\psi) = \phi_x(-\psi^*\psi)=\psi  \]
(note that the first two mappings~$\phi_x$ are the ones defined in this theorem, whereas the third mapping is the one from \eqref{phixy}). We thus obtain
\[ \phi_x(\mycal{R}(\psi,B)) = \Big(\psi, \, \pi_J\big(-\psi^*\psi + \begin{pmatrix} 0 & 0 \\ 0 & B \end{pmatrix} + \psi^*\psi\big)\big|_J \Big) = (\psi, B)\:.\]

Next, the mappings~${\mycal{R}}$ and~$\Phi$ are Fr{\'e}chet-smooth because for operators of finite rank (namely rank at most $2n$),
the operator norm is equivalent to the Hilbert-Schmidt norm.
Indeed, for an operator~$A \::\: H \rightarrow I$ mapping to a finite-dimensional Hilbert space~$I$,
\[ \|A\|^2 \leq \| A^\dagger A \| \leq \tr (A^\dagger A) = \|A\|_{\Shil(\H, I)}^2 \leq \dim(I)\: \|A\|^2 \:. \]
This concludes the proof.
\QED

We consider a smooth curve
\[ \gamma \::\: (-\delta, \delta) \rightarrow \F^\reg \qquad \text{with} \qquad \gamma(0)=x \:.\]
\[ \frac{d}{d\tau} \big(\phi_y \circ \gamma(\tau) \big)\big|_{\tau=0} = \bv \in V_y \:. \]
The corresponding equivalence class defines a tangent vector~$[x,\bv,y] \in T_x \F^\reg$.
On the other hand, considering~$\gamma$ as a curve in~$\Shil$, it has the tangent vector
\[ \frac{d \gamma(\tau)}{d\tau} \Big|_{\tau=0} \in \Shil \:. \]
In the chart~$\phi_x$ and setting~$\psi_0 = \phi_x(x)$, the curve is parametrized by~$\psi(\tau)
:= \phi_x \circ \gamma(\tau)$ with
\[ \gamma(\tau) = \phi_x^{-1} \circ \psi (\tau)=-\psi(\tau)^* \psi(\tau) \]
and thus
\[  \frac{d \gamma(\tau)}{d\tau} \Big|_{\tau=0} = D\phi_x^{-1}|_{\psi_0} \bv= -\bv^* \psi_0 - \psi_0^* \bv \qquad \text{with} \qquad \bv \in V_x \:. \]
As $\psi_0=\phi_x(x)=\pi_x$, a direct computation (for details see the proof of Lemma~\ref{Lemg_xinnerproduct}
in Appendix~\ref{appsymm}) that the map $V_x \ni \bv \mapsto -\bv^*\psi_0-\psi^*_0\bv= -\bv^*\pi_x-\pi^*_x \bv$ is injective.This makes it possible to write the tangent space as
\beq \label{Txid}
T_x \F^\reg \simeq T_x^\Shil \F^\reg := \big\{ -\psi^* \psi_0 - \psi_0^* \psi \:\big|\: \psi \in \Symm(S_x) \oplus \Lin(J, I) \big\}  \subseteq \Shil(\H) \:.
\eeq

\begin{Thm} \label{thmriemann} Using the identification~\eqref{Txid}, the mapping
\[ 
g_x \::\: T_x^\Shil \F^\reg \times T_x^\Shil \F^\reg \rightarrow \R \:,\qquad
g_x(A,B) := \tr(AB) \:. \]
defines
a Fr{\'e}chet-smooth Riemannian metric on~$\F^\reg$. Moreover, the topology on~$\F^\reg$ induced by the operator norm coincides with the topology induced by the Riemannian metric.
\end{Thm}
\Proof Follows immediately because~$g_x$ is
the restriction of the Hilbert space scalar product to the smooth Fr{\'e}chet submanifold~$\F^\reg$.
\QED

We finally remark that the symmetric wave charts are related to Gaussian charts
(see the formulas in~\cite[Sections~5 and 6.2]{gaugefix}, which apply to the infinite-dimensional case
as well). Detailed computations for the Riemannian metric in symmetric wave charts are given in
Appendix~\ref{appsymm}.

\section{Differential Calculus on Expedient Subspaces} \label{secexpedient}
If all functions arising in the analysis were Fr{\'e}chet-smooth, all the methods and notions
from the finite-dimensional setting could be adapted in a straightforward way to the
infinite-dimensional setting. However, this procedure is not sufficient for our purposes,
because the Lagrangian is not Fr{\'e}chet-smooth.
Therefore, we need to develop a differential calculus on Banach spaces for functions
which are only H\"older continuous.
Clearly, in general such functions are not even Fr{\'e}chet-differentiable,
but the G{\^a}teaux derivative may exist in certain directions.
The disadvantage of G{\^a}teaux derivatives is that the differentiable directions
in general do not form a vector space. As a consequence, the usual computation rules
like the linearity of the derivative or the chain and product rules cease to hold.
Our strategy for preserving the usual computation rules is to work on
suitable linear subspaces of the star-shaped set of all G{\^a}teaux-differentiable directions,
referred to as the {\em{expedient differentiable subspace}}.

\subsection{The Expedient Differentiable Subspaces}

In this section $E$ and $F$ denote Banach spaces.
\begin{Def} \label{defV}
Let~$U \subseteq E$ be open and~$f : U \rightarrow F$ an $F$-valued function.
Moreover, let~$V$ be a subspace of~$E$.
The function~$f$ is $k$ times {\bf{$V$-differentiable at~$x_0 \in U$}}
if for every finite-dimensional subspace~$H \subseteq V$,
the restriction of~$f$ to the affine subspace~$H+x_0$ denoted by
\[ g^H : H \rightarrow F \:,\qquad g^H(h) = f(x_0+h) \]
is $k$-times continuously differentiable at~$h=0$.
If this condition holds, the subspace~$V$ is called {\bf{$k$-admissible}} at~$x_0$.
\end{Def}
Thus a function~$f$ is once $V$-differentiable at~$x_0$ if for every finite-dimensional subspace~$H \subseteq V$,
for every~$h_0$ in a small neighborhood of the origin,
\[ g^H(h) = g^H(h_0) + Dg^H|_{h_0} (h-h_0) + o(h-h_0) \qquad \text{for all~$h \in H$} \:, \]
and if~$Dg^H|_{h_0}$ is continuous in the variable~$h_0$ at~$h_0=0$.
Equivalently, choosing a basis~$e_1, \ldots, e_L$ of~$H$, this condition can be stated that
all partial derivatives
\[ \frac{\partial}{\partial \alpha_i} g^H \big( \alpha_1 e_1 + \cdots + \alpha_L e_L \big) \]
exist and are continuous at~$\alpha_1,\ldots, \alpha_L=0$.
The higher differentiability of~$g^H$ can be defined inductively or, equivalently, by demanding
that all partial derivatives up to the order~$k$, i.e.\ all the functions
\[ \frac{\partial^p}{\partial \alpha_{i_1} \cdots \alpha_{i_p}} g^H(\alpha_1 e_1 + \cdots + \alpha_L e_L) \]
with~$i_1,\ldots, i_p \in \{1,\ldots, L\}$ and~$p \leq k$,
exist and are continuous at~$\alpha_1,\ldots, \alpha_L=0$.
%
%
%

An admissible subspace~$V$ is {\em{maximal}} if there are no admissible
proper extensions~$\tilde{V} \supsetneq V$.
The existence of maximal admissible subspaces is guaranteed by Zorn's lemma, but maximal
subspaces are in general not unique. In order to obtain a canonical subspace, we take
the intersection of all maximal admissible subspaces:
\begin{Def} \label{defexpedient}
The {\bf{expedient $k$-differentiable subspace}}~$\E^k(f,x_0)$ of~$f$ at~$x_0$
is defined as the intersection
\[ \E^k(f,x_0) := \bigcap \big\{ V \:\big|\: \text{\rm{$V \subseteq E$ $k$-admissible at~$x_0$ and maximal}} \big\} \:. \]
\end{Def}

Since the expedient differentiable subspace is again admissible at~$x_0$, we obtain a corresponding derivative
as follows. Given~$k\in \N$ and vectors~$h_1, \ldots, h_k \in \E(f,x_0)$, we choose~$H$ as
a finite-dimensional subspace which contains these vectors.
We set
\beq \label{Dfdef}
D^{k,\E} f|_{x_0}(h_1,\ldots, h_k)  := D^k g^H|_0(h_1, \ldots, h_k)
\eeq
(where again $g^H(h):=f(x_0+h)$).
\begin{Lemma} This procedure defines~$D^{k,\E} f|_{x_0}$ canonically as a symmetric, multilinear
mapping
\[ D^{k,\E} f|_{x_0} \::\:
\underbrace{\E^k(f,x_0) \times \cdots \times \E^k(f,x_0)}_{\text{$k$ factors}} \rightarrow F \:. \]
\end{Lemma}
\Proof In order to show that $D^{k,\E} f|_{x_0}$ is well-defined, let~$H$ and~$\tilde{H}$ be two finite-dimensional subspaces of~$\E(f,x_0)$ which
contain the vectors~$h_1, \ldots, h_k$. Then, expressing the partial derivatives in terms
of partial derivatives, it follows that
\begin{align*}
D^k g^H|_0(h_1, \ldots, h_k) &= \frac{\partial^p}{\partial \alpha_1 \cdots \alpha_k} f(x_0 + \alpha_1 h_1 + \cdots + \alpha_k h_k) 
\Big|_{\alpha_1=\cdots=\alpha_k=0} \\
&=  D^k g^{\tilde{H}}|_0(h_1, \ldots, h_k) \:.
\end{align*}
This shows that the definition~\eqref{Dfdef} does not depend on the choice of~$H$.

The symmetry and homogeneity follow immediately from
the corresponding properties of~$D^k g^H$ in~\eqref{Dfdef}.
In order to prove additivity, we let~$h_1, \ldots, h_k \in \E^k(f,x_0)$
and~$\tilde{h}_1, \ldots, \tilde{h}_k \in \E^k(f,x_0)$.
We let~$H$ be the span of all these vectors and use that the corresponding
operator~$D^k g^H|_0$ in~\eqref{Dfdef} applied to~$h_1+\tilde{h}_1, \ldots, h_k +\tilde{h}_k$
 is multilinear.
 \QED
Note that the operator~$D^{k,\E} f|_{x_0}$ is in general not bounded. Moreover, 
$\E^k(f,x_0)$ will in general not be
a closed subspace of~$E$, nor will it in general be dense.

\subsection{Derivatives Along Smooth Curves}
We now analyze under which assumptions directional derivatives exist.
To this end, we let~$I$ be an interval and~$\gamma : I \rightarrow E$ a smooth curve
(here the notions of Fr{\'e}chet and G{\^a}teaux smoothness coincide).
Moreover let~$t_0 \in I$ with~$x_0:=\gamma(t_0) \in U$ and~$U\subseteq E$ open.
Given a function~$f : U \rightarrow F$, we consider the composition
\[ f \circ \gamma \::\: I \rightarrow F \:. \]

\begin{Prp} {\bf{(chain rule)}} \label{prpchain}
Assume that~$f$ is locally H\"older continuous at~$x_0$, meaning that
there is a neighborhood~$V \subseteq U$ of~$x_0$ as well as constants~$\alpha, c>0$ such that
\beq \label{fhoelder}
\| f(x) - f(x') \|_F \leq c\: \|x-x'\|_E^\alpha \qquad \text{for all~$x,x' \in V$}\:.
\eeq
Moreover, assume that all the derivatives of~$\gamma$ at~$x_0$ up to the order
\beq \label{pdef}
p := \bigg\lceil \frac{1}{\alpha} \bigg\rceil
\eeq
(where~$\lceil \cdot \rceil$ is the ceiling function)
lie in the expedient differentiable subspace at~$x_0$, i.e.
\[ \gamma^{(n)}(t_0) \in \E(f,x_0) \qquad \text{for all~$n \in \{1, \ldots, p\}$} \:. \]
Then the function~$f \circ \gamma$ is differentiable at~$t_0$ and
\[ (f\circ \gamma)'(t_0) = D^{\E}f|_{x_0}\, \gamma'(t_0) \:. \]
\end{Prp}
\Proof We consider the polynomial approximation of~$\gamma$
\beq \label{gammapdef}
\gamma_p(t) := \sum_{n=0}^p \frac{\gamma^{(n)}(t_0)}{n!}\: (t-t_0)^n \:.
\eeq
By assumption, this curve lies in the affine subspace~$\E(f,x_0)+x_0$.
Using that the restriction of~$f$ to this subspace is continuously differentiable, it follows that
\[  (f\circ \gamma_p)'(t_0) = D^\E f|_{x_0}\, \gamma'(t_0) \:. \]

It remains to control the error term of the polynomial approximation. Using that~$f$ is
locally H\"older continuous, we know that
\[ \big\| (f\circ \gamma)(t) - (f\circ \gamma_p)(t) \big\|_F
\leq c\: \|\gamma(t) - \gamma_p(t)\|_E^\alpha \:. \]
Using that~$\gamma$ is smooth, it follows that
\beq \label{error}
\big\| (f\circ \gamma)(t) - (f\circ \gamma_p)(t) \big\|_F
\leq \big\| o \big( (t-t_0)^p \big) \big\|_E^\alpha 
= o \big( (t-t_0)^{\alpha p} \big) \:.
\eeq
According to~\eqref{pdef}, we know that~$\alpha p \geq 1$. Therefore,
the error term is of the order~$o(t-t_0)$, which shows that also the function~$t\mapsto (f\circ \gamma)(t) - (f\circ \gamma_p)(t)$ is differentiable with vanishing derivative. This proves the desired result.
\QED

This result immediately generalizes to higher derivatives:

\begin{Prp} {\bf{(higher order chain rule)}} \label{prpchain2}
Assume that~$f$ is locally H\"older continuous at~$x_0$ (see~\eqref{fhoelder}).
Moreover, assume that all the derivatives of~$\gamma$ at~$x_0$ up to the order
\beq \label{pdef2}
p := \bigg\lceil \frac{q}{\alpha} \bigg\rceil
\eeq
lie in the expedient differentiable subspace at~$x_0$, i.e.
\[ \gamma^{(n)}(t_0) \in \E^q(f,x_0) \qquad \text{for all~$n \in \{1, \ldots, p\}$} \:. \]
Then the function~$f \circ \gamma$ is $q$-times differentiable at~$t_0$,
and the derivative can be computed with the usual product and chain rules
(formula of Fa{\`a} di Bruno).
\end{Prp}
\Proof We again consider~$f$ along the polynomial approximation~$\gamma_p$~\eqref{gammapdef}
of the curve~$\gamma$.
By assumption, this curve lies in a finite-dimensional subspace of the affine space
\[ \E^q(f,x_0)+x_0 \;\subset\; F \:. \]
Using that the restriction of~$f$ to this subspace is continuously differentiable,
we know that~$f\circ \gamma_p$ is $q$ times continuously differentiable at~$t=t_0$,
and the derivatives can be computed with the formula of Fa{\`a} di Bruno,
\begin{align*}
(f\circ \gamma_p)^{(q)}(t_0) &= D^{\E,q} f|_{x_0} \big(\gamma'(t_0), \ldots, \gamma'(t_0) \big) \\
&\quad\, + \frac{q(q-1)}{2}\: D^{\E,q-1} f|_{x_0} \big(\gamma''(t_0), \gamma'(t_0), \ldots, \gamma'(t_0) \big) + \cdots \:.
\end{align*}

Using~\eqref{error} and~\eqref{pdef2}, we conclude that
\[ (f\circ \gamma)(t) - (f\circ \gamma_p)(t) = o \big( (t-t_0)^q \big) \:. \]
It follows that also this function is $q$-times differentiable and that all its derivatives vanish.
This concludes the proof.
\QED

\section{Application to Causal Fermion Systems in Infinite Dimensions} \label{secnosmooth}
\subsection{Local H\"older Continuity of the Causal Lagrangian}
The goal of this section is to prove the following result.
\begin{Thm} \label{thmhoelder}The Lagrangian is locally H\"older continuous in the sense that
for all~$x,y_0 \in \F$ there is a neighborhood~$U \subseteq \F$ of~$y_0$ and a constant~$c>0$ such that
\beq \label{Lages}
\big| \L(x,y) - \L(x,\tilde{y}) \big| \leq c\: \|y-\tilde{y}\|^{\frac{1}{2n-1}} \qquad
\text{for all~$y,\tilde{y} \in U$}\:,
\eeq
where~$n$ is the spin dimension. Moreover, the integrand of the boundedness constraint is
locally Lipschitz continuous in the sense that
\beq \label{xyes}
\Big| |x y|^2 - |x\tilde{y}|^2 \Big| \leq c\: \|y-\tilde{y}\|^{\frac{1}{2n}} \qquad
\text{for all~$y,\tilde{y} \in U$}\:.
\eeq
\end{Thm}

We begin with a preparatory lemma.
\begin{Lemma} {\bf{(H\"older continuity of roots)}} \label{lemma-hoelder}
Let
\[ {\mathcal{P}}(\lambda) := \lambda^g + c_{g-1}\, \lambda^{g-1} + \cdots + c_0
= \prod_{i=1}^g (\lambda - \lambda_i) \]
be a complex monic polynomial of degree~$g$ with roots~$\lambda_1, \ldots, \lambda_g$. 
Then there are constants~$C, \varepsilon>0$ such that any
complex monic polynomial~$\tilde{\mathcal{P}}(\lambda) = \lambda^g + \tilde{c}_{g-1}\, \lambda^{g-1} + \cdots + \tilde{c}_0 $ of degree~$g$ which is close to~${\mathcal{P}}$ in the sense that
\[ \|\tilde{\mathcal{P}} - {\mathcal{P}}\| := \max_{\ell \in \{0,\ldots,g-1\} } \big| \tilde{c}_\ell - c_\ell \big| < \varepsilon \]
can be written as~$\tilde{\mathcal{P}}(\lambda) = \prod_{i=1}^g (\lambda - \tilde{\lambda}_i)$
with
\[ 
|\lambda_i - \tilde{\lambda}_i| \leq C\, \|\tilde{\mathcal{P}} - {\mathcal{P}} \|^\frac{1}{p_i}
\qquad \text{for all~$i=1,\ldots, g$}\:, \]
where~$p_i$ is the multiplicity of the root~$\lambda_i$.
\end{Lemma} \noindent
This lemma is proven in a more general context in~\cite[Theorem~2]{brink}.
For self-consistency we here give a simple proof based on Rouch{\'e}'s theorem:

\Proof[Proof of Lemma~\ref{lemma-hoelder}] After the rescaling~$\lambda \rightarrow \nu \lambda$
and~$\lambda_i \rightarrow \nu \lambda_i$ with~$\nu>0$, we can assume that
all the roots~$\lambda_i$ are in the unit ball. Then the polynomial~$\Delta {\mathcal{P}} :=
\tilde{\mathcal{P}} - {\mathcal{P}}$ is bounded in the ball of radius two by
\beq \label{boundDP}
|\Delta {\mathcal{P}}(\lambda)| \leq g\:2^g\: \|\Delta {\mathcal{P}}\| \qquad \text{for all~$\lambda$ with~$|\lambda| \leq 2$}\:.
\eeq
We denote the minimal distance of distinct eigenvalues by
\[ D := \min_{\lambda_i \neq \lambda_j} |\lambda_i - \lambda_j| \:. \]
Since there is a finite number of roots, it clearly suffices to prove the lemma for one of them.
Given~$i \in \{1, \ldots, g\}$ we choose
\beq \label{deldef}
\delta = \bigg( \frac{g\:2^{2g-p_i+1}}{D^{g-p_i}}\: \|\Delta {\mathcal{P}} \| \bigg)^{\frac{1}{p_i}} \:.
\eeq
Next, we choose~$\varepsilon$ so small that~$\delta<D/2$.
We consider the ball~$\Omega = B_\delta(\lambda_i)$. Then for any~$\lambda \in \partial \Omega$,
the polynomial~${\mathcal{P}}$ satisfies the bound
\[ | {\mathcal{P}}(\lambda) | \geq (D/2)^{g-p_i} \,\delta^{p_i} \geq g\:2^{g+1}\: \|\Delta {\mathcal{P}} \|
> |\Delta {\mathcal{P}}(\lambda) | \:, \]
where we used~\eqref{deldef} and~\eqref{boundDP}.
Therefore, Rouch{\'e}'s theorem (see for example~\cite[Theorem~10.36]{rudin})
implies that the polynomials~${\mathcal{P}}$ and~$\tilde{\mathcal{P}}$ have
the same number of roots in the ball~$\Omega$. Thus, after a suitable ordering of the
roots,
\[ |\lambda_i - \tilde{\lambda}_i| \leq \delta \:. \]
Using~\eqref{deldef} gives the result.
\QED

\Proof[Proof of Theorem~\ref{thmhoelder}] Let~$x, y \in \F$.
Since both operators~$x$ and~$y$ vanish on the orthogonal complement of the span their images combined,~$J := \text{span}(S_x, S_y)$, it suffices to compute the eigenvalues on the finite-dimensional
subspace~$J$. Choosing an orthonormal basis of~$S_x=x(\H)$ and extending it to an orthonormal basis of~$J$,
the matrix~$x y|_J- \1_J$ has the block matrix form
\[ \begin{pmatrix} x y \pi_x - \lambda\1 & *  \\
0 & -\lambda \1 \end{pmatrix} \:. \]
Therefore, its characteristic polynomial is given by
\[ \det\nolimits_J (x y- \1_J) = (-\lambda)^{\dim J - \dim x(\H)} \det\nolimits_{x(\H)} \big(x y \pi_x - \lambda \1_{x(\H)} \big) \:. \]

This consideration shows that it suffices to analyze the operators~$x y \pi_x$
and similarly~$x \tilde{y} \pi_x$ on the finite-dimensional Hilbert space~$x(\H)$.
We denote the corresponding characteristic polynomials 
by~${\mathcal{P}}$ and~$\tilde{\mathcal{P}}$, respectively.
They are monic polynomials of degree~$g:= \dim x(\H)$.
The difference of these polynomials can be estimated in terms of operator norms on~$\Lin(\H)$ as follows,
\[ \|\tilde{\mathcal{P}} - {\mathcal{P}}\| \leq c\big(g, \|x\|, \|y\| \big)\: \big\| x \tilde{y} \pi_x - x y \pi_x \big\|
\leq c'\big(g, \|x\|, \|y\| \big) \: \big\| \tilde{y}-y \big\| \:, \]
valid for all~$\tilde{y}$ with~$\|\tilde{y}\| \leq 2 \,\|y\|$.
According to Lemma~\ref{lemma-hoelder}, for sufficiently small~$\|y-\tilde{y}\|$ the eigenvalues of these matrices can be
arranged to satisfy the inequalities
\[ |\lambda_i - \tilde{\lambda}_i| \leq C\, \|\tilde{\mathcal{P}} - {\mathcal{P}} \|^\frac{1}{p_i}
\leq C'\big(x, y \big) \: \big\| \tilde{y}-y \big\|^\frac{1}{p_i} \:. \]

In order to prove~\eqref{xyes}, we consider the estimate
\beq \label{xy2} \begin{split}
\Big| |x y|^2 - |x \tilde{y}|^2 \Big| &\leq \sum_{i=1}^g \Big| |\lambda_i|^2 - |\tilde{\lambda}_i|^2 \Big| \\
&\leq \sum_{i=1}^g |\lambda_i - \tilde{\lambda}_i| \: \big( |\lambda_i| + |\tilde{\lambda}_i| \big)
\leq \tilde{C}(x, y) \: \big\| \tilde{y}-y \big\|^\frac{1}{g}
\end{split}
\eeq
and use that~$g \leq 2n$.

It remains to prove~\eqref{Lages}. In the case~$g<2n$, a simple estimate similar to~\eqref{xy2}
gives the result. In the remaining case~$g=2n$,
using the abbreviation~$\Delta \lambda_i := \tilde{\lambda}_i - \lambda_i$, we obtain
\begin{align*}
\big| \L(x, \tilde{y}) - \L(x,y) \big| &\leq \frac{1}{g} \sum_{i,j=1}^g \Big| |\tilde{\lambda}_i - \tilde{\lambda_j} |^2 -  |\lambda_i - \lambda_j |^2 \Big| \\
&\leq \frac{1}{g} \sum_{i,j=1}^g \Big( 2\:  |\Delta \lambda_i - \Delta \lambda_j |\:|\lambda_i - \lambda_j |
+ |\Delta \lambda_i - \Delta \lambda_j |^2 \Big) \\
&\leq c_2(x, y) \sum_{i,j=1}^g \Big( \big\| \tilde{y}-y \big\|^{\max\big(\frac{1}{p_i}, \frac{1}{p_j} \big) }\:|\lambda_i - \lambda_j | + \big\| \tilde{y}-y \big\|^\frac{2}{g} \Big) \\
&\leq c_3(x, y) \sum_{i,j=1}^g \Big( \big\| \tilde{y}-y \big\|^{\frac{1}{g-1}} + \big\| \tilde{y}-y \big\|^\frac{2}{g} \Big) \:,
\end{align*}
where in the last step we used that, whenever~$\lambda_i \neq \lambda_j$, the multiplicities of
both roots are at most~$g-1$. The inequality
\[ \frac{2}{g} = \frac{1}{n} \geq \frac{1}{2n-1} = \frac{1}{g-1} \:, \]
yields the desired H\"older inequality with exponent~$1/(2n-1)$.
Finally, it is clear from the construction that the constant depends continuously on~$y$.
This concludes the proof.
\QED

In the case of spin dimension one, the Lagrangian is Lipschitz continuous,
in agreement with the findings in~\cite{support}.
If the spin dimension is larger, one still has H\"older continuity, but the H\"older exponent
becomes smaller if the spin dimension is increased. This can be understood from the fact that
the higher the spin dimension is, the higher the degeneracies of the eigenvalues of~$xy$ can be.

We next prove a global H\"older continuity result.
\begin{Thm} {\bf{(Global H\"older continuity)}} \label{thmhoelderglobal}
There is a constant~$c(n)$ which depends only on the spin dimension such that
for all $x,y \in \F$ with~$y \neq 0$ there is a neighborhood $U\subseteq \F$ of~$y$ with
\begin{flalign}
\label{GlobalHoelder}
	| \L(x,y) - \L(x,\tilde{y}) | \leq c(n) \,\|y\|^{2-\frac{1}{2n-1}}\, \|x\|^2\, \| \tilde{y}-y\|^{\frac{1}{2n-1}}
	\qquad \text{for all ~$\tilde{y}\in U$}\:.
\end{flalign}
\end{Thm}
\Proof
Without loss of generality we can assume that $x \neq 0$.
Moreover, using that both sides of the inequality~\eqref{GlobalHoelder} have the same scaling behavior
under the rescaling
\begin{flalign*}
	x \rightarrow \frac{x}{\|x\|}\;, \quad y \rightarrow \frac{y}{\|y\|}\;, \quad\tilde{y} \rightarrow \frac{\tilde{y}}{\|y\|}\;,
\end{flalign*}
it suffices to consider the case that $\|x\|=\|y\|=1$.

Next, choosing a fixed~$4n$-dimensional subspace of $I \subseteq \H$, we can always find a unitary transformation $U: \H \rightarrow \H$ such that $UxU^{-1}(\H), UyU^{-1}(\H) \subseteq I$.
Since the Lagrangian and the operator norms are invariant under such joint unitary transformations (as they leave the eigenvalues of $xy$ invariant), we can assume that both $x$ and $y$ map into the fixed finite dimensional
subspace~$I$.

After these transformations, the operators~$x$ and~$y$ can be considered as operators in~$\Lin(I)$.
Therefore, they lie in the compact set~$\overline{B_1(0)} \subseteq \Lin(I)$.
Since the H\"older constant for the local H\"older continuity depends continuously on $x$ and $y$,
a compactness argument shows that we can choose the H\"older constant uniformly in~$x$ and~$y$: As the previous arguments show, the local H\"older constant can be written as a continuous function $c: \Lin(I)\times \Lin(I) \rightarrow \R^+,\, (x,y) \mapsto c(x,y)$. Since $\overline{B_1(0)} \times \overline{B_1(0)} \subseteq \Lin(I)\times \Lin(I)$ is compact, the local H\"older constant function $c$ is bounded on this set by a constant $c_{\mathrm{max}}>0$, which can then be taken as the desired global H\"older constant.
\QED

\begin{Remark} {\em{
\bitem
\item[{\rm{(1)}}] Since the Lagrangian is symmetric, Theorem~\ref{thmhoelderglobal}
also gives rise to global H\"older continuity with respect to the other argument. Thus
 for all $x,y \in \F$ with~$x \neq 0$ there is a neighborhood $U \subseteq \F$ of $x$ such that
		\begin{flalign}
		\label{GlobalHoelderx}
			|\L(x,y) - \L(\tilde{x},y)| \leq c(n) \|x\|^{2-\frac{1}{2n-1}} \|y\|^2 \|\tilde{x}-x\|^{\frac{1}{2n-1}}
			\qquad \text{for all ~$\tilde{x}\in U$}\:.
		\end{flalign} 
\item[{\rm{(2)}}] As explained in the proof of Theorem~\ref{thmhoelderglobal},
the Lagrangian $\L(x,y)$ depends only on the nonzero eigenvalues of $xy$ and these coincide with the
eigenvalues of $xy\pi_x$. Thus denoting
\[ 
J:=\mathrm{span}(S_x,S_{\tilde{x}}) \:, \]
we immediately obtain the following strengthened version of~\eqref{GlobalHoelderx}:
Every~$x\neq 0$ has a neighborhood~$U \subset \F$ such that the inequality
		\begin{flalign}
		\label{GlobalHoelderxRefined}
		\begin{split}
			|\L(x,y) - \L(\tilde{x},y)| &= |\L(x,\pi_J \,y\,\pi_J) - \L(\tilde{x},\pi_J \,y\, \pi_J)| \\
			&\leq c(n)\: \|x\|^{2-\frac{1}{2n-1}} \:\|\pi_J \,y\, \pi_J\|^2 \:\|\tilde{x}-x\|^{\frac{1}{2n-1}}
		\end{split}
		\end{flalign}
holds for all~$\tilde{x} \in U$ and all~$y \in \F$.
This estimate will be needed for the proof of the chain rule for the integrated Lagrangian~$\ell$
in Theorem~\ref{thmlchain}.
\item[{\rm{(3)}}] In the case~$y=0$, a direct estimate of the eigenvalues shows that
one has H\"older continuity with the improved exponent two,
\[ \big| \L(x,\tilde{y}) \big| \leq c(n) \, \|x\|^2\, \| \tilde{y}\|^2 \:. \]
This inequality can be combined with the result of Theorem~\ref{thmhoelderglobal} to
the statement that for all $x,y$ there is a neighborhood $U\subseteq \F$ of~$y$ with
\beq \label{GlobalHoelder2}
	| \L(x,y) - \L(x,\tilde{y}) | \leq c(n, y) \|x\|^2\, \| \tilde{y}-y\|^{\frac{1}{2n-1}}
	\qquad \text{for all~$\tilde{y}\in U$}\:.
\eeq
Likewise, \eqref{GlobalHoelderxRefined} generalizes to
\beq
\label{GlobalHoelderxRefined2}
|\L(x,y) - \L(\tilde{x},y)| \leq c(n,x) \:\|\pi_J \,y\, \pi_J\|^2 \:\|\tilde{x}-x\|^{\frac{1}{2n-1}} \:.
\eeq
This inequality will be used in the proof of Theorem~\ref{thmlchain}.
\QEDrem
\eitem }}
\end{Remark}

\subsection{Definition of Jet Spaces} \label{secjet}
For the analysis of causal variational principles, the jet formalism was developed
in~\cite{jet}; see also~\cite[Section~2]{fockbosonic}.
We now generalize the definition of the jet spaces to
causal fermion systems in the infinite-dimensional setting.
Our method is to work with the expedient subspaces, where
for convenience derivatives at~$x$ are always computed in the corresponding chart~$\phi_x$.
For example, for analyzing the
differentiability of a real-valued function~$f$ at a point~$x \in \F^\reg$, we consider the composition
\[ f \circ \phi_x^{-1} \::\: \Omega_x \subseteq \Symm(S_x) \oplus \Lin(J,I) \rightarrow \R \:. \]

We introduce~$\Gdiff_\rho$ as the linear space of all vector fields for which the
directional derivative of the function~$\ell$ exists in the sense of expedient subspaces
(see Definition~\ref{defexpedient}),
\[ \Gdiff_\rho =\Big\{ \bu \in C^\infty(M, T\F^\reg) \;\big|\; \text{$\bu(x) \in \E \big(
\ell \circ \phi_x^{-1}, \phi_x(x) \big)$ for all~$x \in M$} \Big\} \:. \]
This gives rise to the jet space
\beq \label{Jdiffdef}
\Jdiff_\rho := C^\infty(M, \R) \oplus \Gdiff_\rho \;\subseteq\; \J_\rho \:.
\eeq
We choose a linear subspace~$\Jtest_\rho \subseteq \Jdiff_\rho$ with the property
that its scalar and vector components are both vector spaces,
\beq\label{Gammatest}
\Jtest_\rho = \Ctest(M, \R) \oplus \Gtest_\rho \;\subseteq\; \Jdiff_\rho \:,
\eeq
and the scalar component is nowhere trivial in the sense that
\[ 
\text{for all~$x \in M$ there is~$a \in \Ctest(M, \R)$ with~$a(x) \neq 0$}\:. \]
It is convenient to consider a pair~$\fu := (a, \bu)$
consisting of a real-valued function~$a$ on~$M$ and a vector field~$\bu$
on~$T\F^\reg$ along~$M$, and to denote the combination of 
multiplication and directional derivative by
\beq \label{Djet}
\nabla_{\fu} \ell(x) := a(x)\, \ell(x) + \big(D_\bu \ell \big)(x) \:.
\eeq

For the Lagrangian, being a function of two variables~$x,y \in \F^\reg$,
we always work in charts~$\phi_x$ and~$\phi_y$, giving rise to the mapping
\beq \label{Lchart}
\L \circ \big(\phi_x^{-1} \times \phi_y^{-1}\big) =
\L \big(\phi_x^{-1}(.), \phi_y^{-1}(.) \big) \::\: \Omega_x \times \Omega_y \subseteq E \rightarrow \R \:,
\eeq
where~$E$ is the Cartesian product of Banach spaces
\[ 
E := \big( \Symm(S_x) \oplus \Lin(J_x,I_x) \big) \times \big( \Symm(S_y) \oplus \Lin(J_y,I_y) \big) \]
with the norm
\[ \|(\psi_x, \psi_y)\|_E := \max \big( \|\psi_x\|_{\Lin(\H)}, \|\psi_y\|_{\Lin(\H)} \big) \]
(where the subscripts~$x$ and~$y$ clarify the dependence on the base points, i.e.\
$I_x = x(H)$, $J_x = I_x^\perp \subseteq \H$ and similarly at~$y$).
We denote partial derivatives acting on the first and second arguments by subscripts~$1$ and~$2$,
respectively. Throughout this paper, we use the following conventions for partial derivatives and jet derivatives:
\begin{itemize}[leftmargin=2em]
\itemD Partial and jet derivatives with an index $i \in \{ 1,2 \}$, as for example in~\eqref{derex}, only act on the respective variable of the function $\L$.
This implies, for example, that the derivatives commute,
\[ 
\nabla_{1,\fv} \nabla_{1,\fu} \L(x,y) = \nabla_{1,\fu} \nabla_{1,\fv} \L(x,y) \:. \]
\itemD The partial or jet derivatives which do not carry an index act as partial derivatives
on the corresponding argument of the Lagrangian. This implies, for example, that
\[ \nabla_\fu \int_\F \nabla_{1,\fv} \, \L(x,y) \: d\rho(y) =  \int_\F \nabla_{1,\fu} \nabla_{1,\fv}\, \L(x,y) \: d\rho(y) \:. \]
\end{itemize}

\begin{Def} \label{defJvary}
For any~$\ell \in \N_0 \cup \{\infty\}$, the jet space~$\J_\rho^\ell \subseteq \J_\rho$
is defined as the vector space of test jets with the following properties:
\begin{itemize}[leftmargin=2em]
\item[\rm{(i)}] The directional derivatives up to order~$\ell$ exist in the sense that
\begin{align*}
\J^\ell_\rho &\subseteq \Big\{ (b,\bv) \in \J_\rho \:\Big|\: \big(\bv(x), \bv(y) \big) \in \Gamma^\ell_\rho(x,y) \\
&\qquad\quad \text{for all~$y \in M$ and~$x$ in an open neighborhood of~$M \subseteq \F^\reg$} \Big\} \:,
\end{align*}
where
\[ \Gamma^\ell_\rho(x,y) := \E^\ell \Big(\L \circ \big(\phi_x^{-1} \times \phi_y^{-1}\big), 
\big(\phi_x(x), \phi_y(y) \big) \Big) \:. \]
The higher jet derivatives are defined by using~\eqref{Djet} and multiplying out, keeping in mind that the partial derivatives act only on the Lagrangian, i.e.\
\begin{align*}
&\nabla^{p, \E} \L \circ \big(\phi_x^{-1} \times \phi_y^{-1}\big)\big|_{(\phi_x(x), \phi_y(y))}
\Big( \big(\fv_1(x), \fv_1(y) \big), \ldots, \big(\fv_p(x), \fv_p(y) \big) \Big) \\
&:= D^{p, \E} \L \circ \big(\phi_x^{-1} \times \phi_y^{-1}\big)\big|_{(\phi_x(x), \phi_y(y))}
\Big( \big(\bv_1(x), \bv_1(y) \big), \ldots, \big(\bv_p(x), \bv_p(y) \big) \Big) \\
&\quad\; + \big(b_1(x)+b_1(y) \big) \: D^{p-1, \E} \L \circ \big(\phi_x^{-1} \times \phi_y^{-1}\big)\big|_{(\phi_x(x), \phi_y(y))}\\
&\qquad\qquad \times
\Big( \big(\bv_2(x), \bv_2(y) \big), \ldots, \big(\bv_p(x), \bv_p(y) \big) \Big) \\
&\quad\; + \big(b_2(x)+b_2(y) \big)\: D^{p-1, \E} \L \circ \big(\phi_x^{-1} \times \phi_y^{-1}\big)\big|_{(\phi_x(x), \phi_y(y))} \\
&\qquad\qquad \times
\Big( \big(\bv_1(x), \bv_1(y) \big), \big(\bv_3(x), \bv_3(y) \big), \ldots, \big(\bv_p(x), \bv_p(y) \big) \Big) \\
&\quad\; + \cdots + \big(b_1(x)+b_1(y) \big) \cdots \big(b_p(x)+b_p(y) \big)\: \L(x,y)\:.
\end{align*}
\item[\rm{(ii)}] The functions
\begin{align}
&\big( \nabla_{1, \fv_1} + \nabla_{2, \fv_1} \big) \cdots \big( \nabla_{1, \fv_p} + \nabla_{2, \fv_p} \big) \L(x,y) \notag \\
&:= \nabla^{p, \E} \L \circ \big(\phi_x^{-1} \times \phi_y^{-1}\big)\big|_{(\phi_x(x), \phi_y(y))}
\Big( \big(\fv_1(x), \fv_1(y) \big), \ldots, \big(\fv_p(x), \fv_p(y) \big) \Big) \label{derex}
\end{align}
are $\rho$-integrable in the variable~$y$, giving rise to locally bounded functions in~$x$. More precisely,
these functions are in the space
\[ L^\infty_\text{\rm{loc}}\Big( M, L^1\big(M, d\rho(y) \big); d\rho(x) \Big) \:. \]
\item[\rm{(iii)}] Integrating the expression~\eqref{derex} in~$y$ over~$M$
with respect to the measure~$\rho$,
the resulting function~$g$ (defined for all~$x$ in an open neighborhood of~$M$)
is continuously differentiable in the direction of every jet~$\fu \in \Jtest_\rho$, i.e.\
\[ \Gtest_x \subseteq \E(g, x) \qquad \text{for all~$x \in M$}\:. \]
\end{itemize}
\end{Def}

\subsection{Derivatives of~$\L$ and~$\ell$ along Smooth Curves}
In this section we use the chain rule in Proposition~\ref{prpchain}
in order to differentiate the Lagrangian~$\L$ and the function~$\ell$ along smooth curves.

\begin{Thm} \label{thmLchain}
Let~$\gamma_1$ and~$\gamma_2$ be two smooth curves in~$\F^\reg$,
\[ \gamma_1, \gamma_2 \in C^\infty((-\delta, \delta), \F^\reg) \:. \]
Setting~$x=\gamma_1(0)$ and~$y=\gamma_2(0)$, we assume that the tangent vectors
up to the order~$p=2n-1$ denoted by
\begin{align*}
\bv_1^{(1)} &:= (\phi_x \circ \gamma_a)'(0) \:, \ldots, \: \bv_1^{(p)} := (\phi_x \circ \gamma_a)^{(p)}(0) \\
\bv_2^{(1)} &:= (\phi_y \circ \gamma_a)'(0) \:, \ldots, \: \bv_2^{(p)} := (\phi_y \circ \gamma_a)^{(p)}(0)
\end{align*}
are in the expedient differentiable subspace of the Lagrangian, i.e.
\[ \big( \bv^{(1)}_1, \bv^{(1)}_2 \big), \ldots, \big( \bv^{(p)}_1, \bv^{(p)}_2 \big) \in \Gamma_\rho(x,y) \:. \]
Then the function~$\L(\gamma_1(\tau), \gamma_2(\tau))$ is $\tau$-differentiable at~$\tau=0$ and
the chain rule holds, i.e.
\begin{align*}
\frac{d}{d\tau} \L\big(\gamma_1(\tau), \gamma_2(\tau) \big)\big|_{\tau=0}
&= D^{\E} \big(\L \circ \big(\phi_x^{-1} \times \phi_y^{-1}\big)\big)\big|_{(\phi_x(x), \phi_y(y))}
\big(\bv_1, \bv_2 \big) \\
&\equiv \big( D_{1, \gamma_1'(0)} + D_{2, \gamma_2'(0)} \big) \L(x,y) \:.
\end{align*}
\end{Thm}
\Proof We again consider the Lagrangian in the charts~$\phi_x$ and~$\phi_y$, \eqref{Lchart}.
In order to show that this function is locally H\"older continuous on~$E$, we begin with the estimate
\begin{align*}
&\big| \L \big(\phi^{-1}_x(\tilde{\psi}_x), \phi^{-1}_y(\tilde{\psi}_y) \big) - \L(x,y) \big| \\
&\leq \big| \L \big( \phi^{-1}_x(\tilde{\psi}_x), \phi^{-1}_y(\tilde{\psi}_y) \big) - \L \big(\phi^{-1}_x(\tilde{\psi}_x), y \big) \big\|
+ \big\| \L \big(\phi^{-1}_x(\tilde{\psi}_x), y \big) - \L(x,y) \big| \\
&\leq c \Big( \|\phi^{-1}_x(\tilde{\psi}_x) - x\|^\alpha_{\Lin(\H)} +  \|\phi^{-1}_y(\tilde{\psi}_y) - y\|^\alpha_{\Lin(\H)} \Big) \:.
\end{align*}
Noting that the function
\[ \phi_x^{-1}(\tilde{\psi}_x) = -\tilde{\psi}_x^* \tilde{\psi}_x \]
is bilinear and therefore Fr{\'e}chet-smooth, it follows that
\begin{align*}
&\big| \L \big(\phi^{-1}_x(\tilde{\psi}_x), \phi^{-1}_y(\tilde{\psi}_y) \big) - \L(x,y) \big| \\
&\leq c C \Big( \|\tilde{\psi}_x - \psi_x \|^\alpha_{\Lin(\H, I)} +  \|\tilde{\psi}_y - \psi_y\|^\alpha_{\Lin(\H, I)} \Big)
\leq 2 c C \:\big\| (\tilde{\psi}_x, - \tilde{\psi}_y) - (\psi_x, \psi_y) \big\|^\alpha_E \:,
\end{align*}
where~$\psi_x := \phi^{-1}_x(x)$ and~$\psi_y := \phi^{-1}_y(y)$.
This proves local H\"older continuity on~$E$.
Applying
Proposition~\ref{prpchain} gives the result.
\QED
We remark that, using Proposition~\ref{prpchain2}, the above method could be generalized in a straightforward manner to higher derivatives.

\begin{Def}
	We call $\ell$ {\bf{H\"older continuous}} with H\"older exponent~$\alpha$ {\bf{along a smooth curve}} $\gamma: I \rightarrow \F$ (with~$I$ an open interval) if for any $t_0 \in I$ with $x_0 = \gamma(t_0)$ there exists a subspace $E_0 \subseteq  \Symm S_{x_0} \oplus \L(J_{x_0},I_{x_0})$ and $\delta >0$ such that the mapping 
	\begin{flalign*}
		\gamma_{x_0}: (t_0-\delta, t_0+\delta) &\rightarrow E_0\:, \quad t \mapsto \phi_{x_0}\circ \gamma (t) - (\1,0)\;,
	\end{flalign*}
	is well defined and locally H\"older continuous with H\"older exponent~$\alpha$.
\end{Def}

\begin{Thm} \label{thmchaingen}
	Let $\gamma: I \rightarrow \F$ be a smooth curve and $\ell$ H\"older continuous along $\gamma$ with H\"older exponent~$\alpha$. For $t_0 \in I$ with $x_0 = \gamma(t_0)$ we set
	\begin{flalign*}
	\ell_{x_0}: E_0 &\rightarrow \R \:,\quad \ell_{x_0}(x) = \ell \circ \phi_{x_0}^{-1} \big(x+(\1,0) \big) \:.
	\end{flalign*}
	If for any $x_0\in I$ the derivatives of $\gamma_{x_0}$ up to the order $p:=\lceil q/\alpha \rceil $ lie in the expedient differentiable subspace at $x_0$, i.e.\
	\begin{flalign*}
		(\gamma_{x_0})^{(n)}(t_0) \in \E^q\Big(\ell_{x_0}, 0\Big)\quad \mathrm{for\;all\;}n\in \{1, \dots, p\}\;,
	\end{flalign*}
	then the function $\ell \circ \gamma = \ell_{x_0} \circ \gamma_{x_0}$ is $q$-times differentiable at $t_0$. Moreover, the usual product and chain rules hold for~$\ell_{x_0} \circ \gamma_{x_0}$.
\end{Thm}
\Proof
	Applying Proposition~\ref{prpchain2} to $\ell_{x_0}$ and $\gamma_{x_0}$ yields the claim as the assumptions for this theorem are clearly fulfilled.
\QED

We now give a sufficient condition which ensures that~$\ell$ is H\"older continuous along~$\gamma$.
This condition needs to be verified in the applications; see for example~\cite{lagrange-hoelder}.
\begin{Thm} \label{thmlchain}
Let $\gamma$ be a smooth curve in $\F$ with
\begin{flalign*}
\int_M \big\| P(\gamma(\tau), y) \big\|^4 \: \big\|Y^{-1} \big\|^2 \:d\rho(y) < C \qquad
\text{for all~$\tau \in (-\delta,\delta)$} \:,
\end{flalign*}
where~$P(x,y)$ is again the kernel of the fermionic projector~\eqref{Pxydef} and~$Y$ is 
(similar to~\eqref{Xdef}) the invertible operator
\[ 
Y := y|_{S_y} \::\: S_y \rightarrow S_y \:. \]
Then the integrated Lagrangian~$\ell$ defined by~\eqref{elldef}
is H\"older continuous along $\gamma$ with H\"older exponent~$\frac{1}{2n-1}$.
\end{Thm}
\Proof
	The idea of the proof is to integrate the estimate~\eqref{GlobalHoelderxRefined2} over $M$.
	To this end, it is crucial to estimate the factor~$\| \pi_J y \pi_J \|$.
	We let~$(\tilde{\phi}_i)_{i\in 1, \dots m}$ be an orthonormal basis of $J$ and denote the orthogonal projection on $\mathrm{span}(\tilde{\phi}_i)$ by $\pi_i$. Since on the finite-dimensional  vector space $L(J)$ all norms are equivalent, we can work with the Hilbert-Schmidt norm of $\pi_J y\pi_J$, i.e.\ for a suitable constant~$C=C(n)$,
\[ \| \pi_J \,y\, \pi_J \|^2 = \| \pi_J \,y\,Y^{-1}\,y \pi_J \|^2 \leq \| \pi_J \,y\|^2 \,\|Y^{-1}\|^2 \, \|y \,\pi_J \|^2 
=  \| \pi_J \,y\|^4 \,\|Y^{-1}\|^2 \:, \]
where in the last step we used that the norm of an operator is the same as the norm of its adjoint.
Combining this inequality with the estimate
\begin{flalign*}
\big\| \pi_J \,y \,\psi \big\|^2 &\leq \big( \big\| \pi_x \,y \,\psi \big\| + \big\| \pi_{\tilde{x}} \,y \,\psi \big\| \big)^2 \leq
2 \,\big\| \pi_x \,y \,\psi \big\|^2 + 2\,\big\| \pi_{\tilde{x}} \,y\, \psi \big\|^2\;,
\end{flalign*}
we obtain
	\begin{flalign*}
		\| \pi_J \,y\, \pi_J \|^2 &\leq 2\,C(n) \:\Big( \big\| \pi_x \,y\, \psi \big\|^2
		+ \big\| \pi_{\tilde{x}} \,y\, \psi \big\|^2 \Big)^2 \:\big\|Y^{-1} \big\|^2 \\
		&\leq 4\, C(n) \:\big\|Y^{-1} \big\|^2 \:\Big( \big\| \pi_x \,y\, \psi \big\|^4
		+ \big\| \pi_{\tilde{x}} \,y\, \psi \big\|^4 \Big) \\
		&= 4\, C(n) \:\big\|Y^{-1} \big\|^2 \: \Big( \big\| P(x,y)\, \psi \big\|^4 + \big\|P(\tilde{x},y) \,\psi \big\|^4 \Big) \:.
	\end{flalign*}
	Using this estimate when integrating~\eqref{GlobalHoelderxRefined2} over $M$ and noting that $\phi_{x}^{-1}$ is locally Lipschitz (since it is Fr{\'e}chet-smooth) yields the claim.
\QED

\appendix
\section{Properties of the Fr{\'e}chet Derivative} \label{appfrechet}
This appendix lists a set of properties and computation rules for Fr{\'e}chet derivatives which are needed for the direct computations in Appendix~\ref{appsymm}. It turns out that most derivation rules known from the finite dimensional case generalize to Fr{\'e}chet derivatives in a straightforward way.

	\begin{Lemma} {\bf{(Properties of the Fr{\'e}chet derivative)}}
	\label{Properties F-derivative}Let $V,W,Z$ be real normed vector. Then the following Fr{\'e}chet derivative rules hold:\bitem
		\item[{\rm{(i)}}] Let $U \subseteq V$ open and $f: V \rightarrow W$ Fr{\'e}chet-differentiable at $x_0\in U$, then $f$ is continuous at $x_0$ and $Df|_{x_0}$ is well defined.
		\item[{\rm{(ii)}}] Let $f\in \Lin(V,W)$ be linear and bounded, then it is Fr{\'e}chet-smooth at any $x_0 \in V$ and $Df|_{x_0} = f$.
		\item[{\rm{(iii)}}] A continuous bilinear map $B: V \times W \rightarrow Z$ is Fr{\'e}chet-smooth at any $(v,w)\in V \times W$ and
		\begin{flalign*}
		DB|_{(v,w)}(h_v,h_w) = B(v,h_w) + B(h_v,w)\;,\;\;\; \forall \, (h_v,h_w)\in V \times W\;.
		\end{flalign*}
		\item[{\rm{(iv)}}] Chain rule: Let $U_V \subseteq V$ and $U_W \subseteq W$ open, $f: U_V \rightarrow W$, $g: U_W \rightarrow Z$ such that $f(U_V)\subseteq U_W$. If $f$ is Fr{\'e}chet-differentiable at~$x_0\in U_V$ and $g$ in $f(x_0)\in U_W$, then also $g\circ f$ is Fr{\'e}chet-differentiable in $x_0$ and
		\begin{flalign*}
		D(g\circ f)|_{x_0} = Dg|_{f(x_0)} \circ Df|_{x_0}\;.
		\end{flalign*}
		\item[{\rm{(v)}}] Let $W_1,...,W_n$ be real normed vector spaces, $W:=W_1 \times W_2 \times...\times W_n$ the product space, $U\subseteq V$ open and $f=(f_1,...,f_n): U \rightarrow W$ with $f_i: V \rightarrow f_i$ for $i=1,\dots,n$. Then~$f$ is Fr{\'e}chet-differentiable at $x_0\in U$ if and only if each $f_i$ is Fr{\'e}chet-differentiable at $x_0$. Moreover, in this case, we have $Df|_{x_0} = (Df_1|_{x_0},...,Df_n|_{x_0})$.
		\item[{\rm{(vi)}}] Let $U_V\subseteq V$ and $U_W \subseteq W$ be open and $f: U_V \rightarrow U_W$ a homeomorphism with inverse $g: U_W \rightarrow U_V$. If $f$ is Fr{\'e}chet-differentiable at~$x_0\in U_V$ and $g$ is Fr{\'e}chet-differentiable in $y_0=f(x_0)\in U_W$, then $Df|_{x_0}$ is an isomorphism with inverse $Dg|_{y_0}$.
\eitem
\end{Lemma} 
\begin{proof}
\begin{itemize}[leftmargin=2.5em]	\item[\rm{(i)}] See \cite[Prop. 2.2, Chapter 2.2]{coleman}.
\item[\rm{(ii)}] $f$ is clearly Fr{\'e}chet-differentiable with $Df|_{x_0}=f$ for any $x_0 \in U$ as  $\|f(x+h) -f(x) -fh\|_W=0$ for all $x,h \in V$ (see also \cite[pp.~149-150]{dieudonne1} and note that the completeness of the vector spaces is not needed for this result). Moreover, as $Df: U \rightarrow L(V,W)$ is constant it is clear that all higher Fr{\'e}chet-derivatives of $f$ vanish (and in particular $f$ is Fr{\'e}chet-smooth).
\item[\rm{(iii)}] $B$ is Fr{\'e}chet differentiable with the stated Fr{\'e}chet derivative as
	\begin{flalign*}
	&\| B(v+h_v,w+h_w)-B(v,w)-B(v,h_w)-B(h_v,w)\|_Z=\| B(h_v,h_w)\|_Z \\
	&\leq C \|h_v\| \cdot\|h_w\|\leq C \big(\mathrm{max}(\|h_v\|,\|h_w\|)\big)^2\;,
	\end{flalign*}
	for a fixed $C>0$ (as $B$ is continuous and bilinear), see also \cite[p. 149-150]{dieudonne1} (again the completeness in not needed). And since \begin{flalign*}
	DB: V \times W &\rightarrow L(V\times W, L(V\times W, Z))\\
	(v,w)&\mapsto \Big( (h_v,h_w) \mapsto B(v,h_w) + B(h_v,w) \Big)\;,
	\end{flalign*}
	is clearly bounded linear, B is Fr{\'e}chet-smooth due to part (ii).
\item[\rm{(iv)}] See \cite[Theorem 2.1, Chapter 2.3]{coleman}.
\item[\rm{(v)}]  See \cite[pp.~149--151]{dieudonne1} (again completeness is not needed).
\item[\rm{(vi)}]  This follows immediately from the chain rule and part (ii) since
	\begin{flalign*}
	\id_V \overset{(ii)}{=} D(\id_V)|_{x_0} = D(g \circ f)|_{x_0} \overset{(iv)}{=} Dg|_{y_0} \circ Df|_{x_0}\;,
	\end{flalign*} 
	and similarly $\id_W = Df|_{x_0} \circ Dg|_{y_0}$.
\end{itemize}

\vspace*{-1.3em}
\end{proof}

	\begin{Lemma}
	\label{F-derivative conecttinating} Let $V$, $W$ and~$Z$ be real normed spaces, $n\in \mathds{N}$ arbitrary and $U_V \subseteq V, U_W \subseteq W$ open subsets, $f: U_V \rightarrow W$  $n$-times Fr{\'e}chet-differentiable (Fr{\'e}chet-smooth) and $g: U_W \rightarrow Z$ $n$-times Fr{\'e}chet-differentiable (Fr{\'e}chet-smooth) such that $f(U_V)\subseteq U_W$. Then $g\circ f$ is also $n$-times Fr{\'e}chet-differentiable (respectively Fr{\'e}chet-smooth).
\end{Lemma}
\begin{proof}
	We show the result by induction over $n$ following \cite[p. 183]{dieudonne1}:
	The case~$n=1$ follows from the chain rule. Now let $n \geq2$ be arbitrary and suppose that the claim holds for~$n-1$. Then
	the induction hypothesis yields that the mapping~$x \mapsto D(f\circ g)|_{x} = Df|_{g(x)}\circ Dg|_{x}$ is $(n-1)$-times Fr{\'e}chet-differentiable, because $Df$, $g$ and $Dg$ are at least $(n-1)$-times Fr{\'e}chet-differentiable and the operator $\circ$ is even Fr{\'e}chet-smooth (as it is bounded linear). Thus $f\circ g$ is $n$-times Fr{\'e}chet-differentiable.
	
	The smoothness result follows immediately from the result for $n$-times differentiability.
\end{proof}

		The following lemma gives a useful computation rule for higher Fr{\'e}chet derivatives (see also \cite[p. 179, 181]{dieudonne1}):
\begin{Lemma}
	Let $V,W$ be real normed spaces, $U\subseteq V$ open and $f: U \rightarrow W$ $n$-times differentiable. Then for any $x_0 \in U$ and $v_1,\cdots, v_n \in V$,
	\begin{flalign}
	\label{Trick Dio}
	D\big(D^{(n-1)}f|_{.}(v_1,\cdots,v_{n-1})\big)\big|_{x_0}v_n = D^{(n)}f|_{x_0}(v_1,\cdots,v_{n})\:.
	\end{flalign}
	In particular, the map
	$U\ni x \mapsto D^{(n-1)}f|_{x}(v_1,\cdots,v_{n-1})\in W$ is Fr{\'e}chet-differentiable.
\end{Lemma}
\begin{proof}
	We follow the idea of the proof given in \cite[pp.~179, 181]{dieudonne1} and also make use of the symmetry result in Lemma~\ref{SymmF-deriv}.
	We first fix $v_1,\dots,v_n \in V$ and define a linear map by
	\begin{flalign*}
	E_{v_1,\dots,v_{n-1}}: \Lin(V,W)_{n-1} &\rightarrow W\\
	A &\mapsto (\dots((Av_1)v_2)\dots)v_{n-1}\;,
	\end{flalign*}
	which simply inserts all the $v_1,\dots,v_{n-1}$ in an $A\in \Lin(V,W)$. Note that $E_{v_1,\dots,v_{n-1}}$ is clearly bounded linear and thus Fr{\'e}chet-smooth. So if we use the representation of $D^{(n-1)}f$ as element of
	\[\underbrace{\Lin(V,\Lin(V,\dots \Lin(V}_{(n-1)\mathrm{ times}},W)\dots ))\;,\]
	the composition $E_{v_1,\dots,v_{n-1}}\circ D^{(n-1)}f$ is also Fr{\'e}chet-differentiable at any $x_0\in U$ with
	\begin{flalign*}
	\Big( E_{v_1,\dots,v_{n-1}}\circ D^{(n-1)}f \Big)'\Big|_{x_0}v_n &= E_{v_1,\dots,v_{n-1}} \circ D^{n}f|_{x_0}v_n =
	(\dots((D^{n}f|_{x_0}v_n)v_1)\dots)v_{n-1} \\
	&=
	D^{n}f|_{x_0}(v_n, v_1,\dots,v_{n-1}) = D^{n}f|_{x_0}(v_1,\dots,v_{n-1}, v_n)\;,
	\end{flalign*}
	where in the first step we used the chain rule and Lemma~\ref{Properties F-derivative}~(ii). In the second step, we used the definition of $E_{v_1,\dots,v_{n-1}}$, whereas in the third step we re-identified $D^{n}f|_{x_0}$ with the corresponding multilinear mapping from $V^n$ to~$W$.
	Finally, in the last step we used the symmetry of $D^{n}f|_{x_0}$.
\end{proof}
We finally state one last computation rule:
\begin{Lemma}
	\label{F-derivative Trick Af}Let $V$, $W$ and~$Z$ be normed vector spaces, $U\subseteq V$ open, $f:U\rightarrow W$ $n$-times Fr{\'e}chet-differentiable and $A \in \Lin(W,Z)$. Then also the function $A\circ f$ is $n$-times Fr{\'e}chet-differentiable
	and
	\begin{flalign}
	\label{deriv. id}
	D^n(A\circ f)|_{x_0}(v_1,\cdots,v_n)=A\Big( D^n f|_{x_0}(v_1,\cdots,v_n) \Big)\qquad
	\forall \, x_0\in U, v_1,\cdots, v_n \in V\;.
	\end{flalign}
\end{Lemma}
\begin{proof}
	The Fr{\'e}chet-differentiability follows immediately from Lemma~\ref{F-derivative conecttinating},
	using that~$A$ is Fr{\'e}chet-smooth. We show the identity~\eqref{deriv. id} by induction over $n$:
The case~$n=1$ follows immediately by the chain rule and Lemma~\ref{Properties F-derivative}~(ii).
Now let $n\geq 2$ and assume that the statement holds for~$n-1$. 
Using the previous lemma (first step), the induction hypothesis (i.e.~\eqref{deriv. id} for the $(n-1)$-st derivative) in the second step as well as the chain rule, Lemma~\ref{Properties F-derivative}~(ii) and the symmetry of $f^(n)$,
for all $x_0\in U$ and $v_1,\cdots, v_n \in V$ we obtain
	\begin{flalign*}
	D^n(A\circ &f)|_{x_0}(v_1,\cdots,v_n)\overset{\eqref{Trick Dio}}{=}D\Big( D^{(n-1)}(A\circ f)|_{.}(v_1,\cdots,v_{n-1}) \Big)\Big|_{x_0}v_n\\
	\overset{(IH)}{=} &D\Big( A\Big( D^{(n-1)}f|_{.}(v_1,\cdots,v_{n-1})\Big) \Big)\Big|_{x_0}v_n = A\Big( D^n f|_{x_0}(v_1,\cdots,v_{n-1},v_n)\Big) \:.
	\end{flalign*}
\end{proof}

\section{The Riemannian Metric in Symmetric Wave Charts} \label{appsymm}
In this subsection we give a detailed computation of the Riemannian metric introduced in Section~\ref{SecRiemMetric} in terms of the symmetric wave charts. Hereby we adapt the methods in~\cite[Section~4]{gaugefix} to the infinite-dimensional setting.

We begin by defining a distance function on $\F^\reg$ by
\begin{flalign*}
d: \F^\reg \times \F^\reg \rightarrow \mathds{R}^{+}_0\;,\;\;\;
(x,y) \mapsto \sqrt{\mathrm{tr}((x-y)^2)}\;.
\end{flalign*}
The trace operator involved here is well-defined
and can be expressed in any orthonormal basis $(e_i)_{i\in \mathds{N}}$ of $\H$ by
(for details see for example~\cite[Section~30.2]{lax})
\begin{flalign}
\label{formal def. Tr}
\mathrm{tr}(A) = \sum_{i=1}^{\infty} \langle e_i|Ae_i\rangle_{\H}
\qquad \text{for $A \in \Lin(\H)$ of finite rank}\:.
\end{flalign}

Moreover note that $d$ does indeed define a distance function on $\F^\reg$ as for any two $x,y \in \F^\reg$
$d(x,y)=\|x-y\|_{\Shil(\H)}$, where~$\|.\|_{\Shil(\H)}$ denotes the Hilbert-Schmidt norm
(see for example~\cite[Section~XI.6]{dunford2} or~\cite[p. 321--322, 309--310]{wernerFA}).

The following remark is a reminder of a calculation rule for the trace operator acting on operators with finite rank.
\begin{Remark} 
	\label{trace remark} {\em{ Let $A\in \Lin(\H)$ be of finite rank and $V \subseteq \H$ a finite-dimensional subspace $V \subseteq \H$ containing the image of $A$, i.e. $A(\H)\subseteq V$. Moreover, let  $(e_i)_{1\leq i\leq k}$ be an orthonormal basis of $v$ and $(\tilde{e}_i)_{i\in\mathds{N}}$ an orthonormal basis of $V^{\bot}$. Then we obtain an orthonormal basis $(\hat{e}_i)_{i\in \mathds{N}}$ of $\H$ by setting $\hat{e}_i:=e_i$ for $i=1,\cdots, k$ and $\hat{e}_{k+j}:=\tilde{e}_j$ for $j\in \mathds{N}$. Using this basis in~\eqref{formal def. Tr} the trace of $A$ reduces to:
	\begin{flalign*}
	\mathrm{tr}(A) = \sum_{i=1}^k \langle \hat{e}_i|A\hat{e}_i\rangle_{\H} = \sum_{i=1}^k \langle e_i|Ae_i\rangle_{\H}\;.
	\end{flalign*} }} \QEDrem
\end{Remark}

The next lemma is mostly based on~\cite[Satz~VI.5.8]{wernerFA} and states some more properties of the trace operators.

\begin{Lemma} {\bf{(Properties of the trace)}}
\label{trace properties}
\bitem
		\item[\em{(1)}] Linearity: The trace operator $\mathrm{tr}$ is linear.
		\item[\em{(2)}] Boundedness: For a finite-dimensional subspace $V\subseteq \H$ consider the corresponding subspace $V_{\Lin}:=\{A \in \Lin(\H)\, | \, A(\H) \subseteq V \} \subseteq \Lin(\H)$. Then $\mathrm{tr}|_{V_{\Lin}}$ is bounded.
		\item[\em{(3)}] Cyclic Permutation: For $x,y \in \Lin(\H)$ with $x$ of finite rank it holds that:
		\begin{flalign*}
		\mathrm{tr}(xy)=\mathrm{tr}(yx)\;.
		\end{flalign*}
		\item[\em{(4)}] Trace of adjoint: For any $x \in \Lin(\H)$ of finite rank also $x^{\dagger}$ is of finite rank and:
		\begin{flalign*}
		\mathrm{tr}(x^{\dagger})=\overline{\mathrm{tr}(x)}\;.
		\end{flalign*}
\eitem
\end{Lemma}
\begin{proof}
	(i): Follows from the definition of tr by \eqref{formal def. Tr}, see also \cite[Satz VI.5.8 (a)]{wernerFA}.\\
	(iii) and (iv): See \cite[Satz VI.5.8 (c),(b)]{wernerFA}.\\
	(ii): Let $A \in V_{\Lin}$. Then, as explained in Remark~\ref{trace remark}, choosing an orthonormal basis $(e_i)_{1\leq i\leq k}$ of $V$ (so $\dim(V)=k$), we can estimate:
	\begin{flalign*}
	|\mathrm{tr}(A)| = \Big|\sum_{i=1}^k\langle e_i| Ae_i\rangle_{\H}\Big| \leq \sum_{i=1}^k|\langle e_i| Ae_i\rangle|_{\H}\leq \sum_{i=1}^k\|A\|_{\H} =k\|A\|_{\H}\;.
	\end{flalign*}
This concludes the proof.
\end{proof}
In the following lemma we consider differentiability properties of a mapping $E$ which corresponds to the square of the distance function $d$. Later we want to use it to introduce the Riemannian metric as second Fr{\'e}chet-derivative of $E$.
\begin{Lemma}
	\label{deriv. d^2}The mappings:
	\begin{flalign*}
	E: \F^\reg \times \F^\reg \rightarrow \mathds{R}\;,\;\;\; (x,y) \mapsto \mathrm{tr}((x-y)^2)\;,
	\end{flalign*}
	and for any fixed $x \in \F^\reg$:
	\begin{flalign*}
	E_x: \F^\reg \rightarrow \mathds{R}\;,\;\;\; y \mapsto \mathrm{tr}((x-y)^2)\;,
	\end{flalign*}
	are Fr{\'e}chet-smooth. Moreover, for all $x,y\in \F^\reg$ with~$x \in \Omega_y$ and all~$\bu, \bv \in V_y$,
\begin{flalign}
	D \big( E_x \circ \phi_y^{-1} \big)\big|_{\phi_y(x)} &=0 \qquad \mathrm{and} \label{first} \\
	D^2\big( E_x \circ \phi_y^{-1} \big) \big|_{\phi_y(x)}(\bv,\bu)
&=4 \re \Big(\mathrm{tr}(y\phi_y(x) \bv^{\dagger}y\phi_y(x)\bu^{\dagger})+\mathrm{tr}(y\bu \bv^{\dagger}y\phi_y(x)\phi_y(x)^{\dagger})\Big) \,. \notag
\end{flalign}
\end{Lemma}
\begin{proof}
	First we have to show that  $E\circ (\phi_x^{-1}, \phi_y^{-1})$ is Fr{\'e}chet-smooth for all $ x,y \in \F^\reg$.\\
	To this end first consider the following calculation for arbitrary $\varphi\in W_x, \psi \in W_y$:
	\begin{flalign*}
	&E \circ (\phi_x^{-1}, \phi_y^{-1})|_{(\varphi, \psi)}=\mathrm{tr}\big((\varphi^{\dagger}x\varphi-\psi^{\dagger}y\psi)^2\big)\\
	&= \mathrm{tr}(\varphi^{\dagger}x\varphi\varphi^{\dagger}x\varphi) - \mathrm{tr}(\varphi^{\dagger}x\varphi\psi^{\dagger}y\psi) - 
	\mathrm{tr}(\psi^{\dagger}y\psi\varphi^{\dagger}x\varphi) +
	\mathrm{tr}(\psi^{\dagger}y\psi\psi^{\dagger}y\psi)\\
	&= \mathrm{tr}(x\varphi\varphi^{\dagger}x\varphi\varphi^{\dagger}) - \mathrm{tr}(x\varphi\psi^{\dagger}y\psi\varphi^{\dagger}) - 
	\mathrm{tr}(y\psi\varphi^{\dagger}x\varphi\psi^{\dagger}) +
	\mathrm{tr}(y\psi\psi^{\dagger}y\psi\psi^{\dagger})
	\;,\\
	&=
	\mathrm{tr}|_{{S_x}_{\Lin}}(x\varphi\varphi^{\dagger}x\varphi\varphi^{\dagger}) - \mathrm{tr}|_{{S_x}_{\Lin}}(x\varphi\psi^{\dagger}y\psi\varphi^{\dagger}) - 
	\mathrm{tr}|_{{S_y}_{\Lin}}(y\psi\varphi^{\dagger}x\varphi\psi^{\dagger}) +
	\mathrm{tr}|_{{S_y}_{\Lin}}(y\psi\psi^{\dagger}y\psi\psi^{\dagger})
	\;,
	\end{flalign*}
	where in the second step we used the linearity of the trace and in the third step
	the cyclic permutation property (which can be applied as all factors and summands obviously have finite rank). The last line is clearly a sum of composition of Fr{\'e}chet-smooth mappings in $(\varphi,\psi)$, which proves the Fr{\'e}chet-smoothness of $E\circ (\phi_x^{-1}, \phi_y^{-1})$.

	For calculating the Fr{\'e}chet derivative of $E_x$ consider the expansion
	\begin{flalign*}
	E_x \circ \phi_y^{-1}(\psi) &= \mathrm{tr}((x-\psi^{\dagger}y\psi)^2) \\
	&= \mathrm{tr}(x^2) -\mathrm{tr}(x\psi^{\dagger}y\psi) - \mathrm{tr}(\psi^{\dagger}y\psi x) + \mathrm{tr}(\psi^{\dagger}y\psi\psi^{\dagger}y\psi) \\
	&= \mathrm{tr}(x^2) -2\mathrm{tr}(x\psi^{\dagger}y\psi) + \mathrm{tr}(y\psi\psi^{\dagger}y\psi\psi^{\dagger})\\
	&=  \mathrm{tr}|_{{S_x}_{\Lin}}(x^2) -2\mathrm{tr}|_{{S_x}_{\Lin}}(x\psi^{\dagger}y\psi) + \mathrm{tr}|_{{S_y}_{\Lin}}(y\psi\psi^{\dagger}y\psi\psi^{\dagger}) \:,
	\end{flalign*}
	which is again a sum of compositions of Fr{\'e}chet-smooth functions showing that also $E_x \circ \phi_y^{-1}$ is Fr{\'e}chet-smooth.

	Applying the computation rule from Lemma~\ref{F-derivative Trick Af} together with the Fr{\'e}chet derivative rule for bilinear functions in Lemma~\ref{Properties F-derivative}~(iii) (multiple times and together with the chain rule) we obtain:
	\begin{flalign*}
	D\big( E_x \circ \phi_y^{-1} \big) \big|_{\psi}\bv=& -2\mathrm{tr}|_{{S_x}_{\Lin}}(x\bv^{\dagger}y\psi) -2\mathrm{tr}|_{{S_x}_{\Lin}}(x\psi^{\dagger}y\bv) + \mathrm{tr}|_{{S_y}_{\Lin}}(y\bv\psi^{\dagger}y\psi\psi^{\dagger}) \\
	&+ \mathrm{tr}|_{{S_y}_{\Lin}}(y\psi \bv^{\dagger}y\psi\psi^{\dagger})
	+ \mathrm{tr}|_{{S_y}_{\Lin}}(y\psi\psi^{\dagger}y\bv\psi^{\dagger})
	+ \mathrm{tr}|_{{S_y}_{\Lin}}(y\psi\psi^{\dagger}y\psi \bv^{\dagger})
	\;.
	\end{flalign*}
	Using that Lemma~\ref{trace properties}~(iii) and (iv) this simplifies to
	\begin{flalign}
	&D\big( E_x \circ \phi_y^{-1} \big) \big|_{\psi}\bv \notag \\
	&= -2\mathrm{tr}(x\bv^{\dagger}y\psi)-2\mathrm{tr}(x\psi^{\dagger}y\bv)
	+2\mathrm{tr}(\psi\psi^{\dagger}y\bv\psi^{\dagger}y)
	+2\mathrm{tr}(y\psi \bv^{\dagger}y\psi\psi^{\dagger}) \notag\\
	&= 4\cdot \re\Big(\mathrm{tr}|_{{S_y}_{\Lin}}(y\psi\psi^{\dagger}y\psi \bv^{\dagger}) -  \mathrm{tr}|_{{S_y}_{\Lin}}(x\bv^{\dagger}y\psi) \Big)\;.
	\label{eq320} \\
	&= 4 \cdot \re\Big(\mathrm{tr}(\psi^{\dagger}y\psi \bv^{\dagger}y\psi) - \mathrm{tr}(x\bv^{\dagger}y\psi)\Big)
	\label{1. derivative d^2}
	\end{flalign}
	In the case~$\psi^{\dagger}y\psi=\phi_y^{-1}(\psi)=x$ the terms in~\eqref{1. derivative d^2} cancel each other, showing that
	\begin{flalign*}
	D( E_x \circ \phi_y^{-1} \big) \big|_{\phi_y(x)} = 0\;.
	\end{flalign*}
	Moreover, proceeding from \eqref{eq320} a straightforward computation using the properties of the
	Fr{\'e}chet derivative and the trace operator as before gives
	\begin{flalign}
	&D^2\big( E_x \circ \phi_y^{-1} \big) \big|_{\psi}(\bv,\bu) \notag \\
	&=4 \cdot\re \Big(\mathrm{tr}|_{{S_y}_{\Lin}}(y\bu\psi^{\dagger}y\psi \bv^{\dagger}) + \mathrm{tr}|_{{S_y}_{\Lin}}(y\psi \bu^{\dagger}y\psi \bv^{\dagger})+\mathrm{tr}|_{{S_y}_{\Lin}}(y\psi\psi^{\dagger}y\bu \bv^{\dagger})\notag \\
	&\phantom{= 4 \cdot\re \Big(} -\mathrm{tr}|_{{S_x}_{\Lin}}(x\bv^{\dagger}y\bu)\Big) \label{eq. Gaussian} \\
	&= 4 \cdot \re\Big(\mathrm{tr}(\psi^{\dagger}y\psi \bv^{\dagger}y\bu) + \mathrm{tr}(y\psi \bu^{\dagger}y\psi \bv^{\dagger})+\mathrm{tr}(y\psi\psi^{\dagger}y\bu \bv^{\dagger})-\mathrm{tr}(x\bv^{\dagger}y\bu)\Big) \;.
	\end{flalign}
	As for $\psi^{\dagger}y\psi=\phi_y^{-1}(\psi)=x$ the first and the last term cancel each other we obtain
	\begin{flalign*}
	&D^2 \big( E_x \circ \phi_y^{-1} \big) \big|_{\phi_y(x)}(\bv,\bu)\\
	&= 4 \cdot \re\Big(\mathrm{tr}(y\phi_y(x)\bu^{\dagger}y\phi_y(x) \bv^{\dagger})+\mathrm{tr}(y\phi_y(x)\phi_y(x)^{\dagger}y\bu \bv^{\dagger})\Big) \:,
	\end{flalign*}
which concludes the proof.
\end{proof}
\begin{Lemma}
	\label{coord. inv. d^2''} $D^2( E_x \circ \phi_y^{-1})|_{\phi_y(x)}$ is independent of the choice of chart (i.e. the choice of $y$) as long as $y\in \F^\reg$ is chosen such that $x \in \Omega_y$.
	Moreover, for all tangent vector fields $\bv,\bu\in \Upgamma(\F^\reg,T\F^\reg)$ and any $y\in \F^\reg$ with $x \in \Omega_y$
	\begin{flalign}
	\label{d^2 coord. inv.}
	D_{\bv(x)} \Big( D_{\bu(.)}E_x(.) \Big)=D^2 \big( E_x \circ \phi_y^{-1} \big) \big|_{\phi_y(x)}\Big(D\phi_y(\bu(x))  ,D\phi_y(\bv(x))\Big)\;,
	\end{flalign}
	where the derivatives act on the arguments containing a dot.
\end{Lemma}
This Lemma also shows, that the the order of differentiation of $E_x$ with respect to the two vector fields does not matter. The proof shows that this is due to the fact that the first derivative of $E_x$ vanishes.
\begin{proof}
	Let $\bv,\bu \in \Gamma(\F^\reg,T\F^\reg)$ and $x,y\in \F^\reg$ with $x \in \Omega_y$ be arbitrary. As we have seen before, for the first directional derivative we have
	\begin{flalign*}
	D_{\bu(\tilde{x})} E_x= D \big( E_x \circ \phi_y^{-1} \big) \big|_{\phi_y(\tilde{x})}D\phi_y(\bu(\tilde{x}))\;.
	\end{flalign*}
	It follows for the second directional derivative that
	\begin{flalign*}
	&D_{\bv(\tilde{x})} \Big( D_{\bu(.)}E_x(.) \Big) 
	=D^2 \big(E_x \circ \phi_y^{-1} \big) \big|_{\phi_y(\tilde{x})}\Big(D\phi_y(\bv(\tilde{x})),D\phi_y(\bu(x))\Big) \\
	&\qquad+ D \big(E_x \circ \phi_y^{-1} \big) \big|_{\phi_y(\tilde{x})} \:D \Big( D \phi_y \big(\bu(\phi_y^{-1}(.)) \big) \Big) \Big|_{\phi_y(\tilde{x})} \:D\phi_y(\bv(\tilde{x}))\;,
	\end{flalign*}
where we applied the Fr{\'e}chet derivative rule for $\mathds{R}$-bilinear maps
in Lemma~\ref{Properties F-derivative}~(iii) together with the chain rule.
Evaluating this expression at $\tilde{x}=x$, the second summand vanishes in view of~\eqref{first}.
We thus obtain
	\begin{flalign*}
D_{\bv(\tilde{x})} \Big( D_{\bu(.)}E_x(.) \Big) \Big|_{\tilde{x}=x} 
=D^2 \big(E_x \circ \phi_y^{-1} \big) \big|_{\phi_y(x)}\Big(D\phi_y(\bv(x)),D\phi_y(\bu(x))\Big)\:.
	\end{flalign*}
Using the symmetry of the second Fr{\'e}chet derivatives gives the result.
\end{proof}

\begin{Remark}
	\label{Rem 4.2.17} {\em{
	Equation \eqref{d^2 coord. inv.} also shows that $D_{\bv(x)}( D_{\bu(.)}E_x(.))$ only depends on the value of the vector fields $\bu$ and $\bv$ at the point $x$. Moreover, since to arbitrary $x \in \F^\reg$ and $\bu,\bv \in T_x\F^\reg$ one can always find a smooth tangent vector field with $\bv(x)=\bv$, $\bu(x)=\bu$ (for example by using a suitable bump function in a chart around $x$), we can consider the expression
	\begin{flalign}
	\label{Def D_2^2}
	D_{2}^2E_x|_{x}(\bu,\bv):=D^2(E_x \circ \phi_y^{-1})|_{\phi_y(x)}\Big(D\phi_y(\bu),D\phi_y(\bv)\Big)\;,
	\end{flalign}
	as a well defined, coordinate invariant -- in the sense that the right hand side of equation \eqref{Def D_2^2} returns the same values for any $y \in \F^\reg$ with $x \in \Omega_y$ -- and symmetric bilinear form. }}
	\QEDrem
\end{Remark}
Now it seems convenient to compute~\eqref{Def D_2^2} in the cart $\phi_x$. Then we have $\phi_x(x)=\pi_x$ and thus we obtain for any $\bu,\bv \in V_x$:
\begin{flalign*}
&D^2\big( E_x \circ \phi_x^{-1} \big)|_{\phi_x(x)}(\bv,\bu)=4 \cdot\re \Big(\mathrm{tr}(x\pi_x\bu^{\dagger}x\pi_x \bv^{\dagger})+\mathrm{tr}(x\pi_x\pi_x^{\dagger}x\bu \bv^{\dagger})\Big)\\
&=4 \cdot \re\Big(\mathrm{tr}(x\bu^{\dagger}x v^{\dagger})+\mathrm{tr}(x^2\bu \bv^{\dagger})\Big)
=4 \cdot\re \Big(\mathrm{tr}(x\bv x\bu)+\mathrm{tr}( x\bu\bv^{\dagger}x)\Big)\;.
\end{flalign*}
Motivated by this for any $x \in \F^\reg$ we set:
\begin{flalign*}
\tilde{g}_x: V_x \times V_x \rightarrow \mathds{R}\;,\;\;\;
(\bu,\bv) \mapsto 4 \cdot \re\Big(\mathrm{tr}(x\bv x\bu)+\mathrm{tr}(x\bu\bv^{\dagger}x)\Big)\;.
\end{flalign*}
Due to the properties of the trace operator, $\tilde{g}_x$ defines a symmetric, real-valued bilinear form on $V_x$, which is even positive-definite as the following lemma shows:
\begin{Lemma}
	\label{Lemg_xinnerproduct}
	The symmetric bilinear form $\tilde{g}_x$ is positive definite and thus defines a real valued inner product on $V_x$.
\end{Lemma}
\begin{proof}
	Let $\bu\in V_x$ be arbitrary, choose an orthonormal basis $(e_i)_{i=1,\cdots,k}$ of the finite-dimensional vector-space $(S_x+\bu^{\dagger}(S_x))$ and compute:\footnote{As explained in Remark~\ref{trace remark}, the trace operator for the finite-rank operators $x\bu x\bu$, $x\bu\bu^{\dagger}x$ and $\bu^{\dagger}x^2\bu$ can indeed be calculated like that as they all map into $(S_x+\bu^{\dagger}(S_x))$.}
	\begin{flalign*}
	&\qquad \quad\;\;\;\,\mathrm{tr}(x\bu x\bu) = \sum_{i=1}^{k}\langle e_i| x\bu x\bu e_i\rangle_{\H} = \sum_{i=1}^{k}\langle \bu^{\dagger}xe_i| x\bu e_i\rangle_{\H}\;, \\
	&\;\,\Rightarrow \re\big( \mathrm{tr}(x\bu x\bu)\big) = \frac{1}{2}\sum_{i=1}^{k}\Big(\langle \bu^{\dagger}xe_i| xue_i\rangle_{\H} + \overline{\langle \bu^{\dagger}xe_i| x\bu e_i\rangle_{\H} }\Big) \\
	&\phantom{\Rightarrow \re\big( \mathrm{tr}(x\bu \bu^{\dagger}x) \big)}\;= \frac{1}{2}\sum_{i=1}^{k}\Big(\langle \bu^{\dagger}xe_i| x\bu e_i\rangle_{\H} + \langle x\bu e_i| u^{\dagger}xe_i\rangle_{\H}\Big)\;,\\
	&\qquad \quad\;\;\mathrm{tr}(x\bu \bu^{\dagger}x) =\sum_{i=1}^{k}\langle e_i| x\bu \bu^{\dagger}xe_i\rangle_{\H} = \sum_{i=1}^{k}\langle \bu^{\dagger}xe_i| \bu^{\dagger}xe_i\rangle_{\H}\;,\\
	&\qquad \quad\;\;\mathrm{tr}(x\bu\bu^{\dagger}x) = \mathrm{tr}(\bu^{\dagger}x^2\bu) = \sum_{i=1}^{k}\langle e_i| \bu^{\dagger}x^2\bu e_i\rangle_{\H} =\sum_{i=1}^{k}\langle x\bu e_i| x\bu e_i\rangle_{\H}\;,\\
	&\Rightarrow \re\big( \mathrm{tr}(x\bu \bu^{\dagger}x) \big) =\re\Big( \frac{1}{2} \sum_{i=1}^{k}\Big( \langle \bu^{\dagger}xe_i| \bu^{\dagger}xe_i\rangle_{\H} + \langle x\bu e_i| x\bu e_i\rangle_{\H}\Big)\Big) \\
	&\phantom{\Rightarrow \re\big( \mathrm{tr}(x\bu \bu^{\dagger}x) \big)}\;= \frac{1}{2} \sum_{i=1}^{k}\Big( \langle \bu^{\dagger}xe_i| \bu^{\dagger}xe_i\rangle_{\H} + \langle x\bu e_i| x\bu e_i\rangle_{\H}\Big)\;,
	\end{flalign*}
	where in the last step we used that $\langle \bu^{\dagger}xe_i| \bu^{\dagger}xe_i\rangle_{\H}= \|\bu^{\dagger}x \|^2_{\H}$ and $\langle x\bu e_i| x\bu e_i\rangle_{\H}= \|x\bu \|^2_{\H}$ are already real (for $i=1,\dots,k$), so we can leave out the "$\re$". \\
	Combining this we obtain:
	\begingroup\setlength\abovedisplayskip{5pt}
	\begin{flalign*}
	\re\Big( \mathrm{tr}(x\bu x\bu) + \mathrm{tr}(x\bu \bu^{\dagger}x) \Big) &= \frac{1}{2}\sum_{i=1}^{k}\Big(\langle (\bu^{\dagger}x +x\bu)e_i| x\bu e_i\rangle_{\H} + \langle (x\bu+\bu^{\dagger}x)e_i| \bu^{\dagger}xe_i\rangle_{\H}\Big) \\
	&=\frac{1}{2}\sum_{i=1}^{k}\langle (\bu^{\dagger}x +x\bu)e_i| (\bu^{\dagger}x +x\bu)e_i\rangle_{\H} \geq 0\;.
	\end{flalign*}
	\endgroup
	This shows the positive semi-definiteness of $\tilde{g}_x$. Moreover we see that $\tilde{g}_x(\bu,\bu)$ vanishes if and only if
	\begin{flalign*}
	0=\big(\bu^{\dagger}x +x\bu\big)|_{S_x+\bu^{\dagger}(S_x)} \;.
	\end{flalign*}
	But as $(\bu^{\dagger}x +x\bu)$ is obviously selfadjoint and its image is contained in $S_x+\bu^{\dagger}(S_x)$, it vanishes on the orthogonal complement of $S_x+\bu^{\dagger}(S_x)$ anyhow,
	so the previous equation is equivalent to
	\begin{flalign}
	0=\label{Bed Riem Metrik 0}
	\bu^{\dagger}x +x\bu\;.
	\end{flalign}
	Moreover, denoting $\pi_I := \pi_x$ as the orthogonal projection on $S_x = I$ and $ \pi_J$ as the orthogonal projection on $J=(S_x)^{\bot}$, 
	we can write
	\begin{flalign*}
	\bu = \bu\pi_I + \bu\pi_J\;.
	\end{flalign*}
	Plugging this in equation \eqref{Bed Riem Metrik 0} yields:
	\begin{flalign}
	\label{Hilfsgleichung}0 = (\bu\pi_I + \bu\pi_J)^{\dagger}x + x(\bu\pi_I + \bu\pi_J) =\pi_I\bu^{\dagger}x + x \bu\pi_I + \pi_J\bu^{\dagger}x + x \bu\pi_J \;,
	\end{flalign}
	Using $\bu|_I \in \mathrm{Symm}(S_x)$ we conclude
	\begin{flalign*}
	X^{-1}\pi_I\bu^{\dagger}X=X^{-1}(\bu|_I)^{\dagger}X = \bu|_I \;,\;\;\;\Rightarrow \pi_I\bu^{\dagger}X = X\bu|_I\;.
	\end{flalign*}
	As $x$ is selfadjoint this also yields
	\begin{flalign*}
	\pi_I\bu^{\dagger}x = x\bu\pi_I\;.
	\end{flalign*}
	Inserting this in \eqref{Hilfsgleichung} gives:
	\begin{flalign}
	\label{Test equation}
	0 = 2x \bu\pi_I + \pi_J\bu^{\dagger}x + x \bu\pi_J  \:.
	\end{flalign}
	Using a block operator notation for the orthogonal decomposition~$\H= I \oplus^{\bot} J$, this equation can be visualized as
	\begin{flalign*}
	\left( \begin{array}{cc}
	0 & 0\\
	0 & 0 
	\end{array} \right)
	= \left( \begin{array}{cc}
	2x \bu\pi_I & x \bu\pi_J\\
	\pi_J\bu^{\dagger}x  & 0 
	\end{array} \right) \:.
	\end{flalign*}
	This notation can be justified by ``testing'' equation \eqref{Test equation} with $(v,0),(0,w)\in \H=I\oplus ^{\bot} J$ with $v\in I$ and $w\in J$ arbitrary.

	Thus we see that each of the operators $2x\bu\pi_I$, $\pi_J\bu^{\dagger}x$ and $x\bu\pi_J$ must vanish individually. Furthermore as $x|_I$ has full rank and $\bu$ maps into $S_x=I$, this yields
	\begin{flalign*}
	\bu\pi_I = 0\;,\;\;\; \bu\pi_J = 0\;,
	\end{flalign*}	
	and therefore also
	\begin{flalign*}
	\bu=\bu(\pi_I+\pi_J)=0\;.
	\end{flalign*}		
	This proves the positive definiteness of $\tilde{g}_x(\bu,\bu)= \re\Big( \mathrm{tr}(x\bu x\bu) + \mathrm{tr}(x\bu\bu^{\dagger}x) \Big)$. 
\end{proof}
Now we can finally introduce a Riemannian metric on $\F^\reg$:
\begin{Lemma}
	Setting pointwise for any $x\in \F^\reg$:
	\begin{flalign*}
	g_x: T_x\F^\reg \times T_x\F^\reg \rightarrow \mathds{R}\;,\;\;\;
	(\bu,\bv) \mapsto D_2^2E_x|_x(\bu,\bv)\;,
	\end{flalign*}
	we obtain a well defined Riemannian metric on $\F^\reg$.\footnote{Where a Riemannian metric on a Banach manifold is defined just as in the finite-dimensional case but with smoothness with respect to the Fr{\'e}chet derivative.}
\end{Lemma}
\begin{proof}
	First of all $g_x$ is well defined due to Lemma~\ref{coord. inv. d^2''} as explained in Remark~\ref{Rem 4.2.17}.\\
	Moreover, choosing representatives $[x,\bu,x]$, $[x,\bv,x] \in T_x\F^\reg$ we have:
	\begin{flalign*}
	&g_x([x,\bu,x],[x,\bv,x])=D_{2}^2E_x|_{x}([x,\bu,x],[x,v,x])\\
	=&D^2(E_x \circ \phi_x^{-1})|_{\phi_x(x)}\Big(\underbrace{D\phi_x([x,\bu,x])}_{=\bu},\underbrace{D\phi_x([x,v,x])}_{=\bv}\Big)=\tilde{g}_x(\bu,\bv)\;,
	\end{flalign*}
	and since we have already seen that for any $x \in \F^\reg$, $\tilde{g}_x$ defines a symmetric positive-definite bilinear form, so does $g_x$.\\
	Thus it only remains to show that $g$ is Fr{\'e}chet-smooth. But due to the (coordinate invariant) definition of $D_2^2E_x|_x$ in~\eqref{Def D_2^2} this follows immediately from the Fr{\'e}chet-smoothness of $D^2(E_x \circ \phi_y^{-1})$ (for this see Lemma~\ref{deriv. d^2}, in particular equation \eqref{eq. Gaussian}). More precisely, since for any two smooth vector fields $\bu,\bv \in\Gamma(\F^\reg,T\F^\reg)$ for any chart $\phi_y$ with $x \in \Omega_y$ also $D\phi_y\circ \bu\circ \phi_y^{-1}$ and $D\phi_y\circ \bv\circ \phi_y^{-1}$ are smooth,  we have
	\begin{flalign*}
	&g_{\phi_y^{-1}(\psi)}(\bu\circ \phi_y^{-1}(\psi),\bv\circ \phi_y^{-1}(\psi)) = D_{2}^2E_x|_{x}(\bu,\bv)\\
	=&D^2(E_x \circ \phi_y^{-1})|_{\psi}\Big(D\phi_y\big(\bu\circ \phi_y^{-1}(\psi)\big),D\phi_y\big(\bv\circ \phi_y^{-1}(\psi)\big)\Big)\;,\;\;\;\forall \, \psi \in W_y\;,
	\end{flalign*}
	which is Fr{\'e}chet-smooth as composition of Fr{\'e}chet-smooth maps. More precisely, introducing the mappings
	\begin{flalign*}
	B_1: \Lin(V_y,\mathds{R})_2 \times V_y \rightarrow \Lin(V_y, \mathds{R})\;,&\;\;\; (A,\bv)\mapsto A\bv\;,\\
	B_2: \Lin(V_y,\mathds{R}) \times V_y \rightarrow \mathds{R}\;,&\;\;\;
	(A',\bv')\mapsto A'\bv'\;,
	\end{flalign*}
	which are both obviously $\mathds{R}$-bilinear and continuous (and thus Fr{\'e}chet-smooth), we can rewrite the previous equation to
	\begin{flalign*}
	&g_{\phi_y^{-1}(\psi)}(\bu\circ \phi_y^{-1}(\psi),\bv\circ \phi_y^{-1}(\psi)) \\
	=& B_2\Big(B_1\Big((d^2(x,\phi_y^{-1}(.)))^{(2)}|_{\psi}, \,D\phi_y(\bu\circ \phi_y^{-1}(\psi))\Big),\, D\phi_y(\bv\circ \phi_y^{-1}(\psi))\Big)\;,
	\end{flalign*}
	which is now clearly a composition of Fr{\'e}chet-smooth maps.
\end{proof}

\Thanks{{{\em{Acknowledgments:}} We are grateful to Olaf M\"uller, Marco Oppio, Johannes
Wurm and the referee
for helpful discussions. M.L.\ acknowledges support by the Studienstiftung des deutschen Volkes.

\providecommand{\bysame}{\leavevmode\hbox to3em{\hrulefill}\thinspace}
\providecommand{\MR}{\relax\ifhmode\unskip\space\fi MR }
\providecommand{\MRhref}[2]{%
  \href{http://www.ams.org/mathscinet-getitem?mr=#1}{#2}
}
\providecommand{\href}[2]{#2}

\end{document}